\documentclass[11pt,tightenlines,nofootinbib,superscriptaddress]{revtex4-2}

\usepackage[toc,page,header]{appendix}
\usepackage{minitoc}
\usepackage[utf8]{inputenc}
\usepackage[british]{babel}
\usepackage{amsmath,amssymb,amsthm,amscd}
\usepackage{mathtools}
\usepackage{mathrsfs}
\usepackage[active]{srcltx}
\usepackage{pgf,tikz}
\usepackage{mathrsfs}
\usepackage{enumerate}
\usepackage{braket}
\usetikzlibrary{arrows}
\usepackage{dsfont}
\usepackage{paralist}
\usepackage{comment}
\usepackage{enumerate}
\usepackage{natbib}
\usepackage{soul}

\definecolor{navyblue}{rgb}{0.0, 0.0, 0.5}
\usepackage[colorlinks=true, urlcolor=blue,citecolor=purple,anchorcolor=blue, linkcolor=black]{hyperref}
\usepackage[margin=2.6cm,bmargin=3cm,tmargin=3cm]{geometry}
\usepackage[nameinlink,capitalise]{cleveref}

\newcommand{\supp}{\operatorname{supp}}

\newcommand{\tr}{\operatorname{tr}}

\newcommand{\cH}{\mathcal{H}}

\newcommand{\norm}[1]{\left\lVert#1\right\rVert}
\renewcommand{\L}{\mathcal L}

\newcommand{\poly}{\operatorname{poly}}
\newcommand{\polylog}{\operatorname{polylog}}
\newcommand{\refrho}{\rho^*}

\newcommand{\rhofixed}{\rho_\infty}

\renewcommand{\phi}{\varphi}
\renewcommand{\epsilon}{\varepsilon}

\theoremstyle{plain}
\newtheorem{theorem}{Theorem}[section]
\newtheorem{corollary}[theorem]{Corollary}
\newtheorem{lemma}[theorem]{Lemma}
\newtheorem{proposition}[theorem]{Proposition}

\newtheorem{definition}[theorem]{Definition}
\newtheorem{assumption}[theorem]{Assumption}

\newtheorem{condition}[theorem]{Condition}
\newtheorem{problem}[theorem]{Problem}

\theoremstyle{remark}
\newtheorem{remark}[theorem]{Remark}

\newtheorem*{claim*}{Claim}
\newtheorem*{remark*}{Remark}
\newtheorem*{example*}{Example}
\newtheorem*{notation*}{Notation}
\numberwithin{equation}{section}

\makeatletter
\@namedef{subjclassname@2020}{\textup{2020} Mathematics Subject Classification}
\makeatother

\begin{document}

\author{\begingroup
\hypersetup{urlcolor=navyblue}
\href{https://orcid.org/0000-0003-1610-1597}{Emilio Onorati}
\endgroup}
\email[Emilio Onorati ]{emilio.onorati@tum.de}
\affiliation{Zentrum Mathematik, Technische Universit\"{a}t M\"{u}nchen, 85748 Garching, Germany}

\author{\begingroup
\hypersetup{urlcolor=navyblue}
\href{https://orcid.org/0000-0001-7712-6582}{Cambyse Rouz\'{e}
\endgroup}
}
\email[Cambyse Rouz\'{e} ]{cambyse.rouze@tum.de}
 \affiliation{Zentrum Mathematik, Technische Universit\"{a}t M\"{u}nchen, 85748 Garching, Germany}
\affiliation{Inria, T\'{e}l\'{e}com Paris - LTCI, Institut Polytechnique de Paris, 91120 Palaiseau, France}

\author{\begingroup
\hypersetup{urlcolor=navyblue}
\href{https://orcid.org/0000-0001-9699-5994}{Daniel Stilck Fran\c{c}a}
\endgroup}
\affiliation{Univ Lyon, ENS Lyon, UCBL, CNRS, Inria, LIP, F-69342, Lyon Cedex 07, France}
\email[Daniel Stilck Fran\c ca ]{daniel.stilck\_franca@ens-lyon.fr}

\author{\begingroup
\hypersetup{urlcolor=navyblue}
\href{https://orcid.org/0000-0002-6077-4898}{James D. Watson}
\endgroup}
\affiliation{University of Maryland, College Park, QuICS 3353 Atlantic Building, MD 20742-2420, USA
}
\email[James D. Watson ]{jdwatson@umd.edu}

\title[]{Provably Efficient Learning of Phases of Matter via Dissipative Evolutions}

\begin{abstract}

The combination of quantum many-body and machine learning techniques has recently proved to be a fertile ground for new developments in quantum computing. 
Since the pioneering work of Huang et al.~\cite{huang2021provably}, several works have shown that it is possible to classically efficiently predict the expectation values of local observables on all states within a phase of matter using a machine learning algorithm after learning from data obtained by measuring other states in the same phase.
However, existing results are restricted to phases of matter such as ground states of gapped Hamiltonians and Gibbs states that exhibit exponential decay of correlations. 
In this work, we drop this requirement and show how it is possible to learn local expectation values for all states in a phase, where we adopt the Lindbladian phase definition by Coser \& P\'erez-Garc\'ia \cite{Coser_Perez-Garcia_2019}, which defines states to be in the same phase if we can drive one to other rapidly with a local Lindbladian. 
This definition encompasses the better-known Hamiltonian definition of phase of matter for gapped ground state phases, and further applies to any family of states connected by short unitary circuits, as well as non-equilibrium phases of matter, and those stable under an external dissipative interaction. 
Under this definition, we show that $N = O(\log(n/\delta)2^{\polylog(1/\epsilon)})$ samples suffice to learn local expectation values within a phase for a system with $n$ qubits, to error $\epsilon$ with failure probability $\delta$. 
This sample complexity is comparable to that for previous results on learning gapped and thermal phases of matter, and it encompasses, in a unified and streamlined way, all previous results of this nature and more. 
As a complementary result, we show that we can learn families of states which go beyond the Coser \& P\'erez-Garc\'ia definition of phase, and we derive a more general bound on the sample complexity which is dependent on the mixing time between states under a Lindbladian evolution.

\end{abstract}

\maketitle

\newpage

\tableofcontents

\newpage

\section{Introduction}

Understanding quantum many-body systems is a fundamental task in quantum chemistry, solid state physics, and quantum information science.
A huge number of powerful techniques have been brought to bear on this problem, including the development of the density matrix renormalization group, tensor network techniques, Monte Carlo, and with the advent of quantum computing, methods such as the Variational Quantum Eigensolver have become popular, as well as more advanced techniques \cite{white1992density, vidal2008class, peruzzo2014variational, gubernatis2016quantum, cirac2021matrix, cubitt2023dissipative}.
Nonetheless, determining the properties of many-body quantum systems from first principles remains a computationally intractable task. 
While the advent of fault-tolerant quantum computers will hopefully enable us to solve a much wider array of problems, in the mean time we are still confined to classical computation. 

From a complexity theoretic perspective, tasks such as finding ground state energies, learning observables measured on many-body quantum states, or determining the boundaries of a phase are known to be computationally intractable in general \cite{kitaev2002classical, ambainis2014physical, gharibian2019oracle, watson2021complexity,  bravyi2022quantum}.
Thus we cannot expect a solution to such problems from just a description of the interactions between particles, however, we might expect that if we have access to additional information, we could significantly speed up these computations.
Recently new techniques have emerged from classical Machine Learning and have been successfully applied to quantum systems, including tasks such as identifying phases of matter \cite{rodriguez2019identifying,rem2019identifying,Dong_Pollmann_Zhang_2019}, characterising observables on phases of matter~\cite{Biamonte2017,Carrasquilla2017, coopmans2023sample}, and approximating quantum states \cite{gao2017efficient,park2020geometry,barr2020quantum, nomura2021purifying}. These proposals mark a paradigm shift: instead of directly computing physical properties of quantum states from a given class, we first generate some measurement data for a few examples (training stage) on a quantum computer or simulator and then use machine learning methods to extrapolate to unseen states.
However, the demonstration of the effectiveness of these techniques has been largely heuristic and in many cases lacked rigorous justification beyond intuition or empirical success.

In this work, we consider such an approach for the task of learning observables everywhere within a phase of matter, which we call \textit{learning the phase} (under some definition of ``phase of matter'').
The ability to learn a phase of matter in this way is desirable as it allows one to learn properties of states which have not been observed from samples which are different from --- but related to --- the state of interest.
This is particularly helpful if preparing a sample of the state of interest is experimentally difficult or computationally intractable
(or just reducing the overhead of sampling many states).
Additionally, it allows for a classical representation of the entire phase to be stored.

In seminal work by Huang et al., the task of learning to predict local observables of quantum many-body ground states everywhere within a gapped phase of matter was considered,
where the gapped phase is defined as the set of ground states where the ground states belong to a family of Hamiltonians which can be smoothly moved between without the spectral gap closing  \cite{huang2021provably}.
The authors proved that a classical machine learning algorithm, with access to data sampled from states within the quantum phase, could outperform algorithms which took no data from input samples (provided the complexity assumption $RP\neq NP$ holds).
A key step in this work was using classical shadow techniques to obtain a classically compact representation of quantum states, which could then be classically post-processed \cite{Huang2020}.

This first result by Huang et al. \cite{huang2021provably} proves learnability of quantum phases of gapped systems and has a sample complexity that is exponential in the precision and polynomial in the system size. 
Subsequent work addressed some of these limitations. 
The learnability of observables was extended to thermal phases of matter which have exponentially decaying correlations in \cite{Onorati_Rouze_Stilch_Franca_Watson_2023}, which also improved the the number of samples needed was improved to $O(\log(n/\delta)2^{\polylog(1/\epsilon)})$.
Improved scaling for gapped ground state phases was also achieved in \cite{Lewis_Huang_Preskill_2023} which reduced the scaling to $O(\log(n/\delta)2^{\polylog(1/\epsilon)})$.
Further work by \cite{che2023exponentially} demonstrated learnability under assumptions about the continuity of matrix elements of the density matrix with similar sample complexity.
However, learnability of more general phases of matter is unknown, and in particular, phases which have algebraically decaying correlation functions (i.e. a weaker form of decay compared to the exponential decay).

 In this work we seek to bridge this gap. 
 Here we use the definition introduced by Coser \& P\'erez-Garc\'ia in \cite{Coser_Perez-Garcia_2019}, under what we call \textit{Lindbladian Phases of Matter}.
 This definition asserts that two states are in the same phase of matter if one can mix from one to the other rapidly using a time-independent, local Lindbladian.
This definition naturally encompasses many of the properties we expect from a phase of matter, including stability under local perturbations and continuity of local observables within the phase.
This allows us to formalise the idea that we can take a representative state in the phase, and then ``move around'' in the phase while the properties vary smoothly.
Furthermore, the definition is physically motivated in the sense that if two states are mixed rapidly by a Lindbladian, then there exists a physical process which allows us to go from one to the other rapidly, and we can show that the ``long-range'' physics of the two states is similar.

 This characterisation provably encompasses the local unitary definition of phase and the gapped Hamiltonian definition of phase, and is expected to generalise thermal phases (although this has yet to be proved rigorously).
 The Lindbladian definition of phase also naturally applies to mixed states in an analogous way to the local unitary definition used to define quantum phases. 
 Since it makes no restrictions on the decay of correlations or spectral gap, this definition is in principle not limited to ground states of gapped Hamiltonians or thermal states with exponential decay of correlations --- it can apply to gapless phases and thermal states with algebraically decaying correlations.
 Rather than characterising these phases as static properties of a Hamiltonian (or other description), this definition gives an interpretation for phases in terms of relations between states.
 %Rather, the phase is defined by the mixing time between states.

 In this setting, we develop a learning procedure which provably allows us to learn observables efficiently and with high probability everywhere in the phase of matter, defined by this Lindbladian definition. The algorithm takes quantum data and performs a classical shadow procedure to extract classical data, which can then be used to construct an estimator for observables of interest.
 Since all we require are the shadows, we could instead directly use classical data generated from some alternative classical procedure such as a tensor network or Monte Carlo simulation.
 This may be advantageous if we can only perform the experiment on a classical computer and wish to reduce the number of computational experiments we wish to perform.

 To prove our results we employ a combination of techniques. 
 We use Lieb-Robinson bounds to relate temporal mixing times to spatial correlation lengths, allowing us to constrain how the system behaves within a phase.
 We then use techniques from classical shadows to reduce the quantum data we have been given to classical data that can be efficiently processed.
 Finally, we use results from concentration of measure to put bounds on sample efficiency and learnability.

\medskip

 This work is structured as follows. In \cref{Sec:Phase_Definition} we review the Lindbladian definition of phase put forward in \cite{Coser_Perez-Garcia_2019}.
 The rigorous statement of the learning problem and the results we prove are stated in \cref{Sec:Main_Results}.
 \Cref{Sec:Learning_Algorithm} gives the learning algorithm and outlines why it works.
 \Cref{Sec:Other_Phase_Definitions} contains a discussion of how our results relate to some other definitions of phase, and some physical examples.
 Finally, in \cref{Sec:Conclusions} we summarise our results and give a discuss for related and future work.

\section{The Lindbladian Definition of Phase of Matter}
\label{Sec:Phase_Definition}

Recent works~\cite{huang2021provably,Onorati_Rouze_Stilch_Franca_Watson_2023, Lewis_Huang_Preskill_2023} have demonstrated how one can learn particular phases of matter efficiently under some structural assumptions. 
However, these works were constrained to phases of ground or Gibbs states satisfying exponential decay of correlations. 
We will now discuss a definition of phase of matter that will allow us to extend these learning results to other classes of physical systems and (partially) recover previous results under a unified framework.
We call this the \textit{Lindbladian definition of phase of matter}.

%Here we will work with the Lindbladian of \textit{Phase of Matter} definition given by \cite{Coser_Perez-Garcia_2019}.
Loosely speaking, the Lindbladian definition of a phase of matter asserts that two separate states on $n$ qubits are in the same phase if there exists a dissipative, time-independent evolution generated by a geometrically local Lindbladian which maps between these states (or rather within $\epsilon$ trace distance in time $\polylog(n/\epsilon)$).
The intuitive motivation for this is as follows: if there is a phase transition between two states, we should expect global, discontinuous changes to the properties of the system, and hence we expect long-ranged correlations to be instituted across the system at some point when evolving one state into the other. 
Famously, the physics of a system at a phase transition is ``scale invariant'' due to these long-ranged correlations.
But if one state can be evolved to another using a geometrically local evolution in only time $\sim \polylog(n)$, then, due to the locality of the interactions in the Lindbladian, there is not sufficient time to create these long-ranged correlations.
Thus the two states should be in the same phase.
In condensed matter, this is related to the notion that the correlation length should diverge when moving between phases.

 We note that typically one can go from a highly-ordered to disordered state quickly by destroying correlations, but the reverse is not true.
 As an example, consider taking some ordered state and applying the local depolarising channel to all the qudits --- it rapidly mixes to the maximally mixed state, but we should not expect the reverse to be true.
 As such, we say that two states are in the same phase only if there is a rapid Lindbladian evolution between the states going \textit{in both directions}.
 Formally, the definition of a Lindbladian phase of matter is given in \cite{Coser_Perez-Garcia_2019} as:

\begin{definition}[Phase of Matter, Definition 1 of \cite{Coser_Perez-Garcia_2019}] \label{Def:Phase_of_Matter}

We say that a state $\rho_0$ can be driven fast to another state $\rho_1$, and we write $\rho_0\rightarrow \rho_1$, if there exists a dissipative evolution generated by a geometrically local and time-independent Lindbladian $\mathcal{L}_n$ acting on the $n$-qudit system and a ancillary system which preserves the locality of the primary system
\begin{align*}
    \norm{e^{t\mathcal{L}_n}(\rho_0\otimes \omega_0) - \rho_1\otimes \omega_1  }_1 \leq \poly(n)e^{-\gamma t},
\end{align*}
with $\omega_0$ and $\omega_1$ respectively the initial and final states of the ancillas, and for a constant $\gamma$.

We say that two states belong to the same phase if there exist two local Lindbladian evolutions as described above such that $\rho_0 \rightarrow \rho_1$ and $\rho_0 \leftarrow \rho_1$, and in this case we write $\rho_0 	\longleftrightarrow \rho_1$.

\end{definition}
\noindent By an ``ancillary system which preserves the locality of the primary system'', we mean that the ancillary system has a spatial structure, and qudits in the ancillary system are only coupled to nearby qudits in the primary system.

  We leave any further discussion of the details of this definition of phase of matter to \cite{Coser_Perez-Garcia_2019}, and simply take this as the definition of phase that we will work with. Throughout we assume that we are working with open or periodic boundary conditions.
\noindent To make this definition of phase valid between systems of different sizes, we also impose a compatibility condition between states which belong to the same phase of different sized systems.
Roughly speaking, if two states are in the same phase at different system sizes.

\begin{condition}[Compatibility Condition]\label{Condition:Compatability_Condition}
Consider two local Lindbladians $\L,\L'$.
Let $\{\L_S\}_S$ denote a family of Lindbladians such that within $S\subseteq \Lambda$, $\L_S$ acts as $\L$, and outside of $S$ it acts as $\L'$.
Denote $\lim_{t\rightarrow \infty}e^{t \L_S}(\rho_0^S\otimes \omega_0^S)(\rho_0) = \rhofixed^S\otimes \omega_1^S$, where $\omega_0^S, \omega_1^S$ are associated auxiliary states.
Then consider three subsets of the lattice $A\subset R \subset W$ such that $A$ does not contain the boundaries of $R$, and $R$ does not contain the boundaries of $W$. 
Then the Lindbladian satisfies the compatibility condition if for\footnote{This results in this paper also hold for $\gamma = \Omega(1/\polylog(n))$, but since this mixing time does not appear physically, will ignore this subtlety.} $\gamma = \Omega(1)$:
\begin{align*}
   | \tr_{A^c}[e^{t\L_R}(\rhofixed^W\otimes \omega_0^R ) - \rhofixed^R\otimes \omega_1^R  ] |\leq \poly(|R|)e^{-\gamma t}.
\end{align*}
\end{condition}

The definition of phase in \cref{Def:Phase_of_Matter} (along with the compatibility condition, \cref{Condition:Compatability_Condition}) has the benefit of (a) defining phases of matter by relative relations between states (b) allowing us to easily define phases of matter for mixed states, not just pure states.
Additionally, this definition of phase allows us to meaningfully talk about phases for finite sized systems, whereas properties associated with phase transitions only truly happen in the thermodynamic limit (e.g. non-analyticity of observables).

In \cref{Def:Phase_of_Matter}, we state we wish to maintain geometric locality not only in the system of interest, but also in the ancillary system. 
To ensure this, we require that there is a notion of locality on the ancillary system as well as the full system, and that the norm of terms acting on the auxiliary system is bounded by the same bound on the primary system.
 We also ensure that all the interactions are of bounded strength and the size of the auxiliary system is no more than $\poly(n)$ if the system of interest is of size $n$.
 The addition of \cref{Condition:Compatability_Condition} enforces that if we consider a small local area of a steady state at two different sizes, then locally they should look very similar and mix rapidly between each other.

\iffalse
\begin{remark}[Ancillary System, Observation (v) of \cite{Coser_Perez-Garcia_2019}]\label{Remark:Ancillary_System}
We wish to maintain geometric locality not only in the system of interest, but also in the ancillary system. 
To ensure this, we require that there is a notation of locality on the ancillary system as well as the full system.
 We also ensure that all the interactions are of bounded strength and the size of the auxiliary system is no more that $\poly(n)$ if the system of interest is of size $n$.
\end{remark}
\fi 

A key point that follows from the previous definitions and conditions is that generally, local observables will satisfy what is called \emph{local rapid mixing}, which states that local regions of the lattice mix to their steady state in a time independent of the full system size. 
More formally:
\begin{condition}[Local Rapid Mixing]
\label{Condition:LOcal_Rapid_Mixing}
     Let $\L$ be a Lindbladian with steady state $\rho_\infty$, and let $\rho_0$ be some point in the same phase as $\rho_\infty$, such that $\rho_0$ is the steady state of the Lindbladian $\L'$.
 Let $O_A$ be any geometrically local observable supported on $A\subset \Lambda$.
Then:
    \begin{align*}
       | \tr[O_A(T_t(\rho_0) - \rho_\infty)]| \leq \poly(|A|) e^{-\gamma t},
    \end{align*}
    for some constant $\gamma$ determined by the Lindbladian.
\end{condition}

\noindent In \cref{Sec:Learning_Lindbladian_Phases} we demonstrate that local rapid mixing follows from \cref{Def:Phase_of_Matter} with \cref{Condition:Compatability_Condition}. 
%, it is arguably a reasonable assumption as it ensure smoothness of local observables in the thermodynamic limit, and is assumed to be true in the original paper defining Lindbladian phases of matter (see Observation \textit{(vii)} and eq. (3) of \cite{Coser_Perez-Garcia_2019}).

The Lindbladian definition of phase should be compared to the Hamiltonian definition of \textit{phase of matter} defined for pure states, in which a phase can be described either as the set of ground states of a smoothly parameterised family of Hamiltonians for which the gap does not close, or by a constant depth circuit where the unitaries have $\polylog(n)$ locality \cite{chen2010local}. 
These two definitions can be shown to be equivalent as per  \cite{bachmann2012automorphic} or \cite[Sec. 3.1]{Coser_Perez-Garcia_2019}. \
The fundamental motivation between the Hamiltonian definition and the Lindbladian definition is the same --- a gapped family of ground states is guaranteed to have exponentially decaying correlations everywhere, and so long-range correlations are never generated.
Similarly, when considering the finite-depth circuit definition of phase, we realise that time evolution generated by local operators should not generate long-ranged correlations and should ensure the smoothness of the expectation values of local observables.
However, the Hamiltonian and circuit definitions of phase struggle to deal with mixed states, whereas the Lindbladian definition naturally encapsulates them.
Importantly, like the circuit definition for gapped ground state phases, the Lindbladian definition here is applicable to phases with algebraic decay of correlations, rather than the stricter exponential decay of correlations.

\begin{figure}
    \centering
    \includegraphics[width=1.0\textwidth]{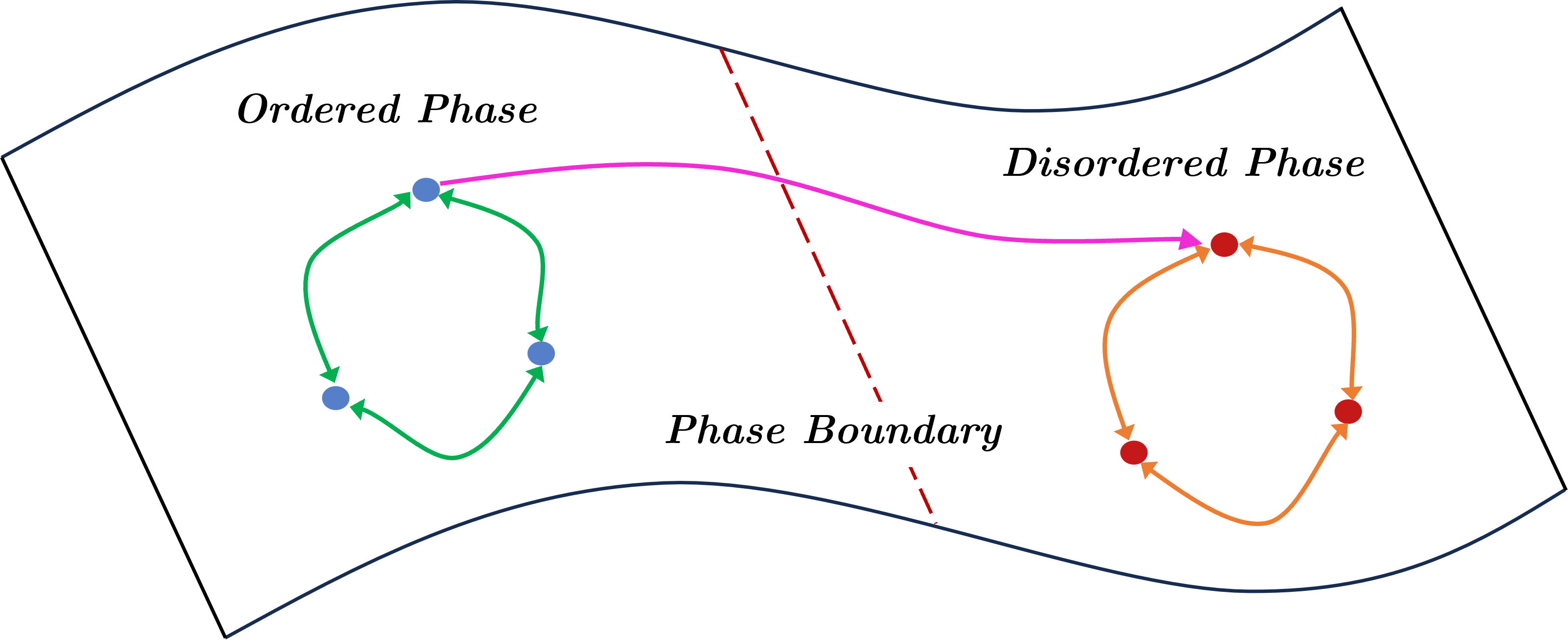}
    \caption{A schematic diagram of states one either side of a phase boundary. 
    Crosses represent states in the respective phases and arrows represent rapid evolutions generated by fixed Lindbladians between states. 
    There does not exist a rapid evolution from the disordered state to the ordered state as we expect that crossing the phase boundary from disorder to order requires the generation of long-ranged correlations.}
    \label{Fig:Phase_Diagram}
\end{figure}

\subsection{Transitivity of Lindbladian Phases}
%\textcolor{red}{STOPPED HERE.}
Although not immediately obvious, the Lindbladian definition of phase divides the set of states into equivalence classes and provides a partial order among these classes (see \cref{Fig:Phase_Diagram} as an example).
The definition is clearly reflexive and symmetric, but it can also be shown to be transitive: if $\rho_0 \rightarrow \rho_1$ and $\rho_1 \rightarrow \rho_2$, then there exists a local, time-independent Lindbladian which takes $\rho_0 \rightarrow \rho_2$ (see \cite[Sec. 3.2]{Coser_Perez-Garcia_2019} for an explicit proof of this). 
Therefore the Lindbladian definition of phase establishes equivalence classes of states --- each phase is such an equivalence class.
Although not proven, intuitively one should think of the ordering of these equivalence classes as physically corresponding to more highly ordered phases --- more ``ordered'' phases should be harder to create. 
For example, with an Ising model Hamiltonian, we should be able to rapidly go from a highly ordered phase (e.g. ferromagnetic phase) to a disordered phase (e.g. high temperature phase), but should not expect the reverse to be true.
Generating the ferromagnetic phase requires spins to be aligned everywhere, and hence requires long-ranged correlations.

\subsection{Parameterising Phases of Matter}

Given a phase of matter that we wish to learn, we need an efficient way of referring to the states in the phase.
With this in mind we arbitrarily choose a state $\refrho$ in the phase to be our reference state, and we will define all other states with respect to this state in terms of the Lindbladian and the time to reach them under that Lindbladian.
From the transitive property of the definition of phase, the choice of reference state does not matter.
As such we can parameterise the states in the phase as:
\begin{align} \label{Eq:Phase_State_Description}
    \rho(\L,\omega, t) \coloneqq e^{t\L}(\refrho \otimes \omega ),
\end{align}
where $t\geq 0$.
Thus the question of parameterising the phase becomes a question of parameterising the Lindbladian.
We allow for each local term in the Lindbladian to be parameterised in both a continuous and discrete manner, such that $ \L^{q}(x) \coloneqq \sum_{j=1}^m \L^{q_j}_j(x_j)$, $x_j\in \mathbb{R}^{\ell}$, $q_j\in [Q]^\ell$, and where $Q$ is discrete labelling of local terms. 
The subset of $\mathbb{R}^\ell$ that $x_j$ is restricted to is determined by the parameters of the phase.
More formally:
\begin{definition}[Efficiently Parameterisable Phase] \label{Def:Efficiently_Parameterised}
    A set of states which form a phase of matter, in the sense of \cref{Def:Phase_of_Matter}, is \emph{efficiently parameterisable} if the set of all Lindbladians describing the phase, in the sense of \cref{Eq:Phase_State_Description}, can be efficiently parameterised. 
    We will assume the set of such Lindbladians can be written as $\{\L^{q}(x)\}_{x,q}$, $x\in \Phi \subseteq \mathbb{R}^m$, $q\in [Q]^m$ with 
    \begin{align*}
        \L^{q}(x) \coloneqq \sum_{j=1}^m \L^{q_j}_j(x_j) ,
    \end{align*}
    where each $\L^{q_j}(x_j)$ is a local term, and $x_j\in \mathbb{R}^{\ell}$, $q_j\in [Q]^\ell$ for $\ell=O(1)$.
    Here $x\in [-1,1]^m$, is a continuous parameterisation of the Lindbladian, and $q\in [Q]^m$ is a discrete labelling of the Lindbladian terms.
    We further assume that the auxiliary state $\omega$ is in a finite set of states $\omega\in \{\omega_i \}_{i=1}^W$.
    Here $m=O(n)$ and $Q, W = O(1)$.
\end{definition}
We note that for a phase of matter to be efficiently learnable, it must at least be efficiently parameterisable, otherwise there is no way of efficiently describing the state that we're interested in learning about. 
Thus restricting to the set of efficiently parameterisable phases is necessary.
There may be other ways of efficiently parameterising Lindbladians compared to \cref{Def:Efficiently_Parameterised}, in which case we expect our results to still hold.
\Cref{Def:Efficiently_Parameterised} should thought of similarly to how a gapped ground state phase can be efficiently described by the parameterising the corresponding set of Hamiltonians.
%\textcolor{red}{James: NEED TO CORRECT THE FOLLOWING:
%\Cref{Def:Efficiently_Parameterised} should be compared  to the description of the Hamiltonian can be used to efficiently parameterise ground states and corresponding thermal states. }

\section{Main Results}
\label{Sec:Main_Results}

The set up is as follows: given a phase of matter as per \cref{Def:Phase_of_Matter} and \cref{Condition:Compatability_Condition}, we wish to predict expectation values of local observables everywhere in the phase with high probability using information from a limited number of samples drawn from points in the phase.
An instance where this may be interesting is if these samples are sets of states that are expensive to prepare in a laboratory, and hence we want to minimise the number that we have to prepare. 
Alternatively, they maybe be the output of a similarly expensive quantum (or classical) computation.
Being able to solve \cref{Problem:Learning_Phases} (formulated below) would allow us to learn the observable $O$ at points in the phase without doing the actual work of preparing the state.
In general, there may be regions of the parameter space that are harder to reach in terms of state preparation or computational effort.

\begin{problem}[Learning Phases of Matter]
\label{Problem:Learning_Phases}
Assume we can parameterise the Lindbladians defining phase of matter as $\L(x), \ x\in \Phi \subseteq \mathbb{R}^m$, where $\Phi$ defines the set of parameters within the phase.
Given $N$ samples $\{\rho(x_i, \omega, \tau_i)\}_{i=1}^N$ of a quantum state, drawn from different points of the parameter-space of a phase of matter from some probability distribution, and a local observable $O=\sum_i O_i$, predict the function
    \begin{align*}
        f_O(\L, \omega ,t) \coloneqq \tr[O \rho(\L, \omega, t) ],
    \end{align*}
everywhere in the phase of matter.
\end{problem}
\noindent Although the problem, as phrased here, only applies to linear functions of the state, it is possible to extend this to consider non-linear functions by taking tensor products.
The phase $\Phi \subseteq \mathbb{R}^m$ is simply a set of parameters describing Lindbladians where the Lindbladian definition of phase holds, e.g. it may be something such as $\Phi = [-1,1]^m$.

\medskip

 We give a learning algorithm, which, if given samples taken from a sufficiently anti-concentrated distribution across the phase, allows us to predict any local observable $O$ with high probability.

\begin{theorem}[Learning Algorithm for Phases of Matter (Informal)]\label{Theorem:Main_Result}
Let $n$ be the system size.
With the conditions of the previous paragraph, given a set of $N$ samples $\{(x_i, \tau_i),\tilde{\rho}(\L( x_i), \tau_i)\}_{i=1}^N$, where $\tilde{\rho}(\L(x_i),\tau_i)$ can be stored efficiently classically, and $N = \mathcal{O} \big(\log\big(\frac{M}{\delta}\big)\,\log\big(\frac{n}{\delta}\big) e^{\operatorname{polylog}(\epsilon^{-1}) } \big)$,
 there exists an algorithm that, on input $x\in\Phi \subseteq \mathbb{R}^m, t \geq 0$ and a local observable $O=
\sum_{i=1}^M O_i$, produces an estimator $\hat{f}_O$
such that, with probability $(1-\delta)$,
\begin{align*}
  \sup_{x\in[-1,1]^m}\, |f_O(\L(x), \omega,t)-\hat{f}_{O}(\L(x), \omega,t)|\le \epsilon\,\sum_{i=1}^M\|O_i\|_\infty\,.
\end{align*}
Moreover, the samples $\tilde{\rho}(\L( x_i), \tau_i)$ are efficiently generated from measurements of the states in the phase $\{\tilde{\rho}(\L( x_i), \tau_i)\}_{i=1}^N$ followed by classical post-processing which takes time $O(nN)$.
\end{theorem} 

\noindent The formal version is given in \cref{Theorem:Learning_General_Phase}. This demonstrates we only need quasi-polynomially many samples in $\frac{1}{\epsilon} $ and logarithmically many in $n$ to learn the entire phase, and the classical post-processing can be done efficiently.
We note that although our results as stated assume some parameterisation of the phase, we do not need to know this parameterisation for the learning algorithm --- we need only the promise that the parameterisation exists and is sufficiently well-behaved.
This allows us to deal with systems where we have imperfect knowledge of the system and state.
For example, if there is noise present, but we do not know its exact form, but may know it is proportional to some parameter (e.g. temperature). 

Finally, we remark that because our method ultimately uses classically stored data, and thus it also works for data which is generated and stored classically. 
For example, data from using Monte Carlo methods for quantum systems (or similarly any other classical model method such as tensor networks).
%This is advantageous when such simulations are computationally expensive.

\subsection{Learning Steady State Phases}

Often the idea of a phase of matter in the sense of \cref{Def:Phase_of_Matter} is too general, and what we are actually interested in is a subset of the general case, for example, a quantum phase of matter (i.e. the ground states of a family of Hamiltonians in the same phase).
Since ground states of Hamiltonians can always be written as a steady state of a local Lindbladian evolution (see \cite{verstraete2009quantum, cubitt2023dissipative}), we can study these phases by first restricting to the subset of states which are the steady states of the Lindbladians\footnote{Thermal states can be written as fixed points of local Lindbladians in certain case, e.g. when they can be prepared as a local quantum circuit \cite{brandao2019finite}. 
Other Lindbladians which prepare Gibbs states e.g. \cite{kastoryano2013quantum,chen2023thermal} are not generally local operators.}.
We call these \textit{Steady State Phases} (see \cref{Fig:Fixed_Point_Phase} for an illustration).
We denote the steady state of a Lindbladian $\L$ as $\rhofixed(\L, \omega)$, or just $\rhofixed(\L)$ where the auxiliary state is not relevant. The following corollary then follows from \cref{Theorem:Main_Result} (for a full proof and statement, see \cref{Sec:Fixed-Point_Phases}).
\begin{figure}[h!]
    \centering
    \includegraphics[width=0.6\textwidth]{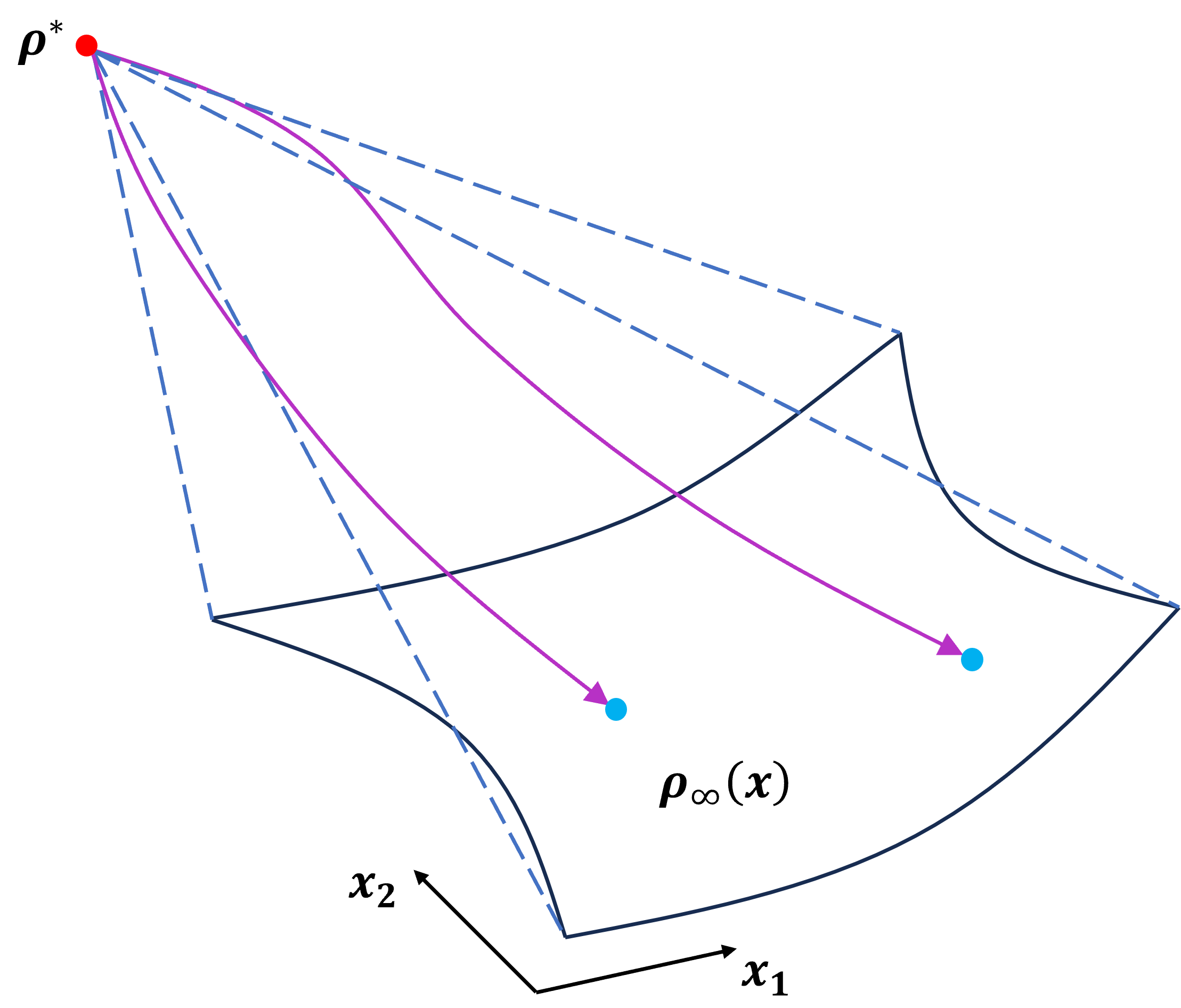}
    \caption{A schematic diagram of a Steady State Phase defined relative to a state $\rho^*$ (red circle). 
    States in the phase are represented by blue circles and the phase is represented by the surface with black outlines (parameterised by $x_1, x_2$).
    The phase is the set of states $\{\rhofixed(x)\}_x$ which are steady states of a set of Lindbladian evolutions from $\rho^*$ (the Lindbladian evolution is denoted by the purple paths).
    }
    \label{Fig:Fixed_Point_Phase}
\end{figure}

\begin{corollary}[Learning Algorithm for Steady State Phases (Informal)]\label{Theorem:Main_Result_Fixed_Points}
With the conditions of the previous paragraph and \cref{Theorem:Main_Result}, given a set of $N$ samples %$\{(x_i),\tilde{\rho}_\infty(\L( x_i))\}_{i=1}^N$, where $\tilde{\rho}_\infty(\L(x_i))$ can be stored efficiently classically, and 
$N = \mathcal{O} \big(\log\big(\frac{M}{\delta}\big)\,\log\big(\frac{n}{\delta}\big) e^{\operatorname{polylog}(\epsilon^{-1})} \big)$,
 there exists an algorithm that, on input $x\in\Phi\subseteq \mathbb{R}^m$ and a local observable $O=
\sum_{i=1}^M O_i$, produces an estimator $\hat{f}_O$
such that, with probability $(1-\delta)$,
\begin{align*}
  \sup_{x\in[-1,1]^m}\, |f_O(\L(x))-\hat{f}_{O}(\L(x))|\le \epsilon\,\sum_{i=1}^M\|O_i\|_\infty\,.
\end{align*}
%Moreover, the samples $\tilde{\rho}_\infty(\L( x_i))$ are efficiently generated from measurements of the states in the phase $\{\tilde{\rho}_\infty(\L( x_i))\}_{i=1}^N$ followed by classical post-processing.
\end{corollary}

%\noindent The samples can be given as density matrices which are then sampled from to get classical data, or the samples can be given directly as classical data.

\subsection{Learning Classes of States Without Rapid Mixing}

We can of course consider learning families of states which do not necessarily satisfy the Lindbladian definition of phase or rapid mixing.
That is, we may be able to relate states by a Lindbladian evolution, but they may not satisfy \cref{Def:Phase_of_Matter} or \cref{Condition:Compatability_Condition}.
Instead, we suppose they satisfy a condition similar to \cref{Def:Phase_of_Matter}, but with a more general \textit{slow mixing condition}: 
\begin{align}\label{Eq:Slow_Mixing}
    \norm{ e^{t\mathcal{L}}(\rho_0\otimes \omega_0) - \rho_1\otimes \omega_1    }_1 \leq f(n)e^{-\gamma t}.
\end{align}
The longer mixing time allowed by this condition means that, as we move between states in the family we are considering, new correlations can potentially be introduced over long distances.
Thus we should expect the behaviour between such states to behave less smoothly.
We show that it is still possible to learn such families of states, but that we require additional resources which increase as the mixing time increases.
\begin{theorem}[Learning under Slow Mixing] \label{Theorem:Learning_Under_Local_Rapid_Mixing}
    Consider a set of states satisfying the slow-mixing assumption \cref{Eq:Slow_Local_Mixing}, then the number of samples required to learn an estimator of local expectation values (in the sense of \cref{Theorem:Main_Result}) scales as:
    \begin{align}
        N = O \left(\log\bigg(\frac{M}{\delta}\bigg)\,\log\bigg(\frac{n}{\delta}\bigg) e^{\operatorname{polylog}(f(n) / \epsilon)} \right).
    \end{align}
\end{theorem}
\noindent We note that we do not recover \cref{Theorem:Learning_General_Phase} in the case $f(n) = \poly(n)$ as we have not also imposed the compatibility condition between systems of different sizes from \cref{Condition:Compatability_Condition}.
That is, this slow mixing learning does not necessarily require compatibility between different system sizes, and can characterise systems where there is unusual or unstable behaviour as the system grows.

\section{The Learning Procedure and Proof Outline}
\label{Sec:Learning_Algorithm}

\subsection{The Learning Procedure}

The method we use to construct an estimator is remarkably simple.
We consider learning steady state phases, but the idea generalises to the general case straightforwardly.
We consider the $N$ samples we are given from various parameters $x_i$, and for each of these samples we perform a randomised 1-local Clifford measurement.
If we record the $x_i$ and the measurement outcome, we can construct a single-measurement classical shadow for point $\widetilde{\rhofixed}(x_j)$.
These single-measurement shadows will be written as:
\begin{align*}
    \widetilde{\rhofixed}(x_j) &= \bigotimes_{i=1}^n(\ket{z_i}\bra{z_i})  \\
    \ket{z}&\in \{ \ket{0}, \ket{1}, \ket{+}, \ket{-}, \ket{+i}, \ket{-i} \}
\end{align*}
where $\ket{\pm i}$ are the eigenstates of the Pauli $Y$ operator.
The eigenstate assigned to each qubit depends on the randomised Pauli measured and the resultant measurement outcome. As explained in~\cite{Huang2020}, it is then possible to use this data to construct very efficient estimators for local properties of quantum states given i.i.d. copies. 
Later, \cite{Onorati_Rouze_Stilch_Franca_Watson_2023} extended this to the case where we are given copies of states that are close to each other. We will use that version here.

To construct our estimator for the observable $O$ on a state $\rhofixed(y)$, for some point $y$ in the phase ($y\in [-1,1]^m$), we consider the parameters $x_i$ of our samples.
We then choose a set of points $\Gamma = \{x_j\}_j$, where a sample is added to $\Gamma$ if $||x_{j|A(r)} - y_{|A(r)}||_{\ell_\infty}\leq \gamma$ for some appropriately small parameter $\gamma$ and some sufficiently large region $A(r)$.
Here $A(r)$ is a ball of radius $r=\Theta(\log(\epsilon^{-1}))$ around the support of the observable $O$, where $\epsilon$ is the precision we wish our estimator to be correct to.
We then construct an estimator:
\begin{align*}
    f_O(y) = \frac{1}{|\Gamma|}\sum_{j\in \Gamma} \tr[O\widetilde{\rhofixed}(x_j)],
\end{align*}
%where $j$ is in the set $\Gamma$ if, on a large area $A(r)$ around the support of $O$, the indices of $x_j$ are sufficiently close to $y$.
%That is $||x_{j|A(r)} - y_{|A(r)}||_{\ell_\infty}\leq \gamma$ for some appropriately small parameter $\gamma$ and some sufficiently large region $A(r)$.
%The area $A(r)$ is chosen to be large enough such that the correlations between $O$ and the area outside of $A(r)$ do not vary much between states in the phase.
%The size of the area is, roughly speaking, the area which correlations between $O$ and the rest of the system could be generated after mixing between states.
A flowchart of the learning and predictions state is given in \cref{Fig:Algorithm_FLowchart}.

\begin{figure*}[h!]
\includegraphics[width=\linewidth]{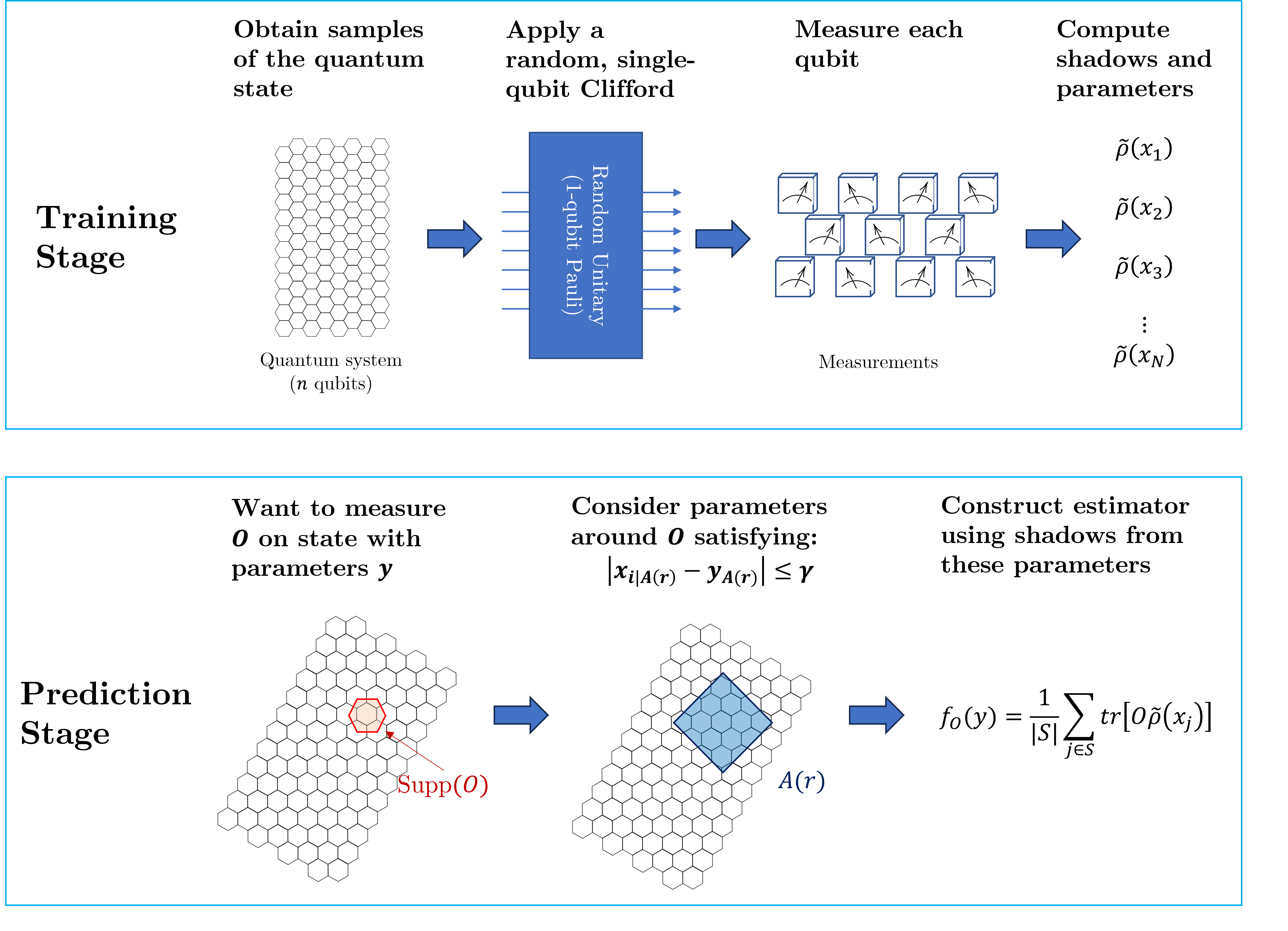}
\caption{Schematic representation of our algorithm to learn phases of matter. The training stage just consists of collecting shadows corresponding to various parameters. In the prediction stage, given an observable $O$ and corresponding parameter $y$ supported on a region $\text{Supp}(O)$, we search for parameters $x^i$ we sampled that have parameters close to $y$  on an enlarged region around $A$ and compute the expectation value of $O$ on the corresponding shadows. The prediction is then a median of means estimate on the values. Note that no machine-learning techniques are required for the estimate.
}
\label{Fig:Algorithm_FLowchart}
\end{figure*}

\subsection{Proof Outline}
Here we give a brief outline of the proof of learning of steady state phases (\cref{Theorem:Main_Result_Fixed_Points}).
The proof for general Lindbladian phases follows in a similar manner.

\ \newline
\textbf{Localising the Physics of the System:} \ 
Consider two steady states $\rhofixed(x), \rhofixed(x')$ of Lindbladians $\L(x),\L(x')$.
If these two points are in the same phase, then there is a rapidly mixing path from a reference state to both of them under $\L(x),\L(x')$ respectively.
The fact $\L(x),\L(x')$ are local Lindbladians means that they have associated Lieb-Robinson bounds --- a finite speed of propagation of information in the system.
In simple terms, Lieb-Robinson bounds say that if $\L_{A(r)}$ is the restriction of $\L$ to the region $A(r)$, then the evolution of an operator at the centre of $A(r)$ follows:
\begin{align*}
    \norm{e^{t\L^*}(O) - e^{t\L_{A(r)}^*}(O) }\leq \norm{O}|\supp(O)|J \frac{(e^{vt}-1-vt)e^{-\beta r}}{v} 
\end{align*}
for appropriate constants $\beta, v$, and where $J$ characterises the strength of the Lindbladian.
The full statement is in \cref{Lemma:Localised_Evolution}.

This, plus the rapid mixing condition, implies that the long-range physics $\rhofixed(x), \rhofixed(x')$ must be similar for both states as the Lindbladian evolution does not have time to change the long-range correlations.
Thus, given an expectation value of some local observable $A$ measured on $\rhofixed(x), \rhofixed(x')$ respectively, we should only expect the difference between $\tr[O\rhofixed(x)], \tr[O\rhofixed(x')]$ to depend on the parameters of $x, x'$ describing the system geometrically close to the support of $O$.

The proof uses techniques from \cite{Cubitt_Lucia_Michalakis_Perez-Garcia_2015} to show that steady states of Lindbladians which are close in parameter space also share similar properties, and look the same on a local section of the lattice.

\ \newline
\textbf{Local Physics Implies More Efficient Local Parameterisation:} \ 
We now want to characterise how large the region around the support of $O$ we need to consider in order to parameterise $\tr[\rhofixed(x)O]$.
We see that the rate of rapid mixing gives the answer to this, and that the effect of parameters away from the support of $O$ decays rapidly, with a spatial decay rate depending on the mixing time.
This tell us that even if $x,x'$ are very different overall, as long as they are close in the region around $O$, then $\tr[\rhofixed(x)O], \ \tr[\rhofixed(x')O] $ will have similar values.

\ \newline
\textbf{Sampling from the Phase and Constructing the Estimator:} \ 
We now assume we are given $N$ quantum states sampled from points $Y_1,\dots, Y_N\sim \mathcal{D}$ from the phase, where $\mathcal{D}$ is a sufficiently anti-concentrated distribution (i.e. that attributes sufficient mass to all points).
For a given observable $O_i$, we then estimate the reduced steady state states over large enough enlargements $S_i\partial $ of the observable supports 
$\mathcal{S}_i:=\{x_j|\,\operatorname{supp}(L_j(x_j))\cap S_i\partial\ne \emptyset\}\cap [x-\epsilon,x+\epsilon]^m$.

%Our primary task is constructing an estimator  $\hat{f}_O$ which is constructed as follows. 
%During a training stage, we pick $N$ points $Y_1,\dots, Y_N\sim U$ and estimate the reduced steady state states over large enough enlargements $S_i\partial $ of the supports 
%$\mathcal{S}_i:=\{x_j|\,\operatorname{supp}(h_j(x_j))\cap S_i\partial\ne \emptyset\}\cap [x-\epsilon,x+\epsilon]^m$ of the observables $O_i$.

Assuming an anti-concentration property of the distribution $\mathcal{D}$ which our samples are selected from, the probability that a small region $\mathcal{S}_i\partial$ in parameter space contains $t$ variables $Y_{i_1},\dots, Y_{i_t}$ becomes large for large enough $N$. 
We then run the classical shadow tomography protocol on those states separately in order to construct efficiently classically describable and computable product matrices $\widetilde{\rhofixed}( Y_1),\dots,\widetilde{\rhofixed}(Y_N)$ \cite{Huang2020}. 
Then for any region of the lattice $S_i$, we select the shadows $\widetilde{\rhofixed}(Y_{i_1}),\dots \widetilde{\rhofixed}(Y_{i_t})$ whose local parameters are close to that of the target state and construct the empirical average $\widetilde{\rhofixed}_{S_i }(x):=\frac{1}{t}\sum_{j=1}^t\,\tr_{S_i^c}\big[\widetilde{\rhofixed}(Y_{i_j})\big]$.

Since we are promised that the states are in the same phase, as per \cref{Def:Phase_of_Matter}, then there must be a rapidly mixing Lindbladian between them, and hence as discussed in the previous paragraph, they must look similar locally.
The estimator $\hat{f}_O$ is then naturally chosen as an empirical average $\hat{f}_O(x):=\sum_{i=1}^M\,\tr[O_i\,\widetilde{\rhofixed}_{S_i }(x)]$.

\section{Relation to Learning from Other Definitions of Phase}
\label{Sec:Other_Phase_Definitions}

\subsection{Gapped Hamiltonian \& Circuit Definitions of ``Phase of Matter''}
\label{Sec:Circuit_Definition}

An important subset of cases we are interested in is the case of ground states of Hamiltonians.
Notably, for gapped ground states, a ``phase of matter''  can be described either as a the set of ground states that can be smoothly moved between in constant time using the adiabatic theorem, or by a constant depth circuit where the unitaries have polylog locality. 
These two definitions can be shown to be equivalent.

Importantly for our purposes, the Lindbladian definition of phase encompasses the circuit definition of phase.
That is, if two states are related by a regular path of Hamiltonians $H(s)$, such that for all $s$ the Hamiltonian is gapped, then these two states also satisfy the Lindbladian definition of phase \cite[Section 3]{Coser_Perez-Garcia_2019}.
Thus we get as a corollary:
\begin{corollary}[Efficient Learning of Gapped Ground State Phases] \label{Corollary:Ground_State_Learning}
    Consider the ground states of a Hamiltonian $H(x)$, where $\Phi \subseteq \mathbb{R}^{m}$, where for $x\in \Phi$, the spectral gap of $H(x)$ is lower bounded by an $O(1)$ constant (i.e. $\Phi$ characterises a gapped phase of matter). 
    Then the family of ground states are in the same phase as per the Lindbladian definition, \cref{Def:Phase_of_Matter}, and therefore can be learned with a number of samples scaling as $N = O( \log(n/\delta)\log(M/\delta)2^{\polylog(1/\epsilon)})$.

    Furthermore, the same scaling holds for phases of matter defined by states which are related by an $O(1)$ depth circuit with unitaries with $O(\log(n))$ locality.
\end{corollary}

We also note that the circuit definition of ``phase of matter'' goes beyond just the set of ground states of gapped Hamiltonians, and can more generally be used to characterise equivalence classes of states \cite{piroli2021quantum}.
Since the above \cref{Corollary:Ground_State_Learning} holds for $\polylog(n)$ depth circuits with $O(1$) local gates, we see that our learning result applies to equivalence classes of states defined by these local unitaries.
Indeed, it allows learning of phases which correspond to gapless phases of matter, provided the states satisfy the above circuit definition of phase of matter.

\subsection{Examples for Rapidly Mixing Systems}

There are some cases where we can relate the Lindbladian definition of phase to physical examples.
It is known that there is a finite depth circuit relating the low-temperature 4D Toric code to its ground state \cite{hastings2011topological}. 
We can map this construction a time-independent Lindbladian to show that the low-temperature Gibbs state and ground state of the 4D Toric code are in the same phase.

More generally, bounding mixing times for Lindbladians is a notoriously difficult task, which makes it hard to put the Lindbladian definition of phase in contact with physical systems (see the introduction of \cite{capel2020modified} for a discussion of this for thermal phases of matter). 
%For example, the 2D Ising model is only known to mix polynomially fast above its critical temperature under Glauber dynamics \cite{lubetzky2012critical}.

We point out that for commuting Hamiltonians in 1D, all states mix rapidly in $\log(n)$ time under an evolution generated by the Davies generator, showing that the Lindbladian definition agrees with the standard definition of phase which says that there are no thermal phase transitions in 1D systems \cite{bardet2023rapid}.
At high-temperatures (i.e. above a phase transition) it is known that one can rapidly mix between Gibbs states of the classical 2D Ising model under Glauber dynamics \cite[Chapter 15]{levin2017markov},
%It is also known that the 2D Ising model rapidly converges to the high temperature phase %in a time independent of size \cite{lubetzky2013cutoff}, 
and similarly if a uniform external magnetic field is applied \cite{martinelli19942}, which also holds true for systems satisfying strong spatial mixing.
%Notably, it has been shown that one can mix to high temperature Ising Model states using standard Glauber dynamics in time $O(n\log(n/\epsilon))$, or parallel Glauber dynamics in time $O(\polylog(n/\epsilon))$ \cite[Theorem 1.2]{lee2023parallelising}.
An $O(\log(n/\epsilon))$ mixing time also holds provided certain conditions on the Gibbs state are satisfied \cite{fischer2018simple}.

However, these authors know of very little literature showing mixing times between states which are in the same, non-trivial phase (e.g. low-temperature thermal states).
\cite{martinelli2003ising} demonstrates that if the boundary conditions are fixed, the Ising model on a tree mixes in time $O(n\log(n/\epsilon))$ under Glauber dynamics (which morally could be considered $O(\log(n/\epsilon))$ under parallelised Glauber dynamics).
Since fixed boundary conditions can be considered long-ranged correlations, this arguably constitutes the states being in the same phase.

Chen et al. \cite{chen2023thermal} demonstrate a Lindbladian which maps to Gibbs states, for which they expect to the spectral gap to be constant for high temperature states, or states in the same phase, hence we expect might expect $\log(n/\epsilon)$ mixing times\footnote{Technically the spectral gap of the Lindbladian does not always characterise the mixing time, but suggests a fast mixing time.}.
However, the Lindbladian here is not local, and our results do not direct apply here.
\cite{brandao2019finite} demonstrate that certain classes of thermal states satisfying that correlations are exponentially decaying, and that erased patches can be recovered locally, can be prepared using log-depth local quantum circuits suggesting these are in the high-temperature phase by the Lindbladian definition.

More generally, we can always start from a reference state and consider the family of states which is generated by evolving that state for a finite time by a family of parametrised local Lindbladians. 
This includes families of states related by finite-depth quantum circuits.
Thus, in principle, the learning results in this work can be applied to gapless phases of matter or a reference state with algebraically decaying correlations where this condition holds.
%\cite{kastoryano2016quantum} demonstrates the rapid mixing for high-temperature and 1D Gibbs states of commuting Hamiltonians, demonstrating they are in the trivial phase by the Lindbladian definition.

%

\section{Conclusions and Discussion}
\label{Sec:Conclusions}

We have shown that it is possible to sample-efficiently learn observables in phases of matter defined under the Lindbladian definition of phase. 
This definition of phase not only covers previously investigated phases of matter such as gapped ground states phases, but allows us to generalise to other phases of matter including gapless phases as matter, phase defined by mixed states, such as those with noise or some contact with an external system, and at least some thermal phases.
The Lindbladian definition of phase is a very natural definition which incorporates intuitive aspects of phase, including \textit{(i)} stability under perturbations \textit{(ii)} similarity of correlations \textit{(iii)} a partial ordering of states.
Despite this generality, our methods allow for sample efficient learning which is similar to \cite{Onorati_Rouze_Stilch_Franca_Watson_2023} and \cite{Lewis_Huang_Preskill_2023}, and better than \cite{huang2021provably}.
Furthermore, the estimators constructed here are simple empirical averages, and are both conceptually and computationally easy to implement. Thus, given both its generality, power to encompass a variety of physical systems and technical simplicity, we believe that the framework laid out in this work has the potential to find applications in a variety of settings at the intersection of machine learning and quantum many-body systems.

\subsection{Future Work and Related Questions}

\noindent \textbf{Lindbladian Phases in Physical Systems.}
Currently, there are still some gaps in our understanding of which families of states constitute a phase under the Lindbladian definition of phase. 
Although we can construct Lindbladians for which the gapped ground state phases satisfy the rapid mixing condition, in other cases (e.g. thermal phases) this is not so clear.
Greater investigation into this definition is warranted, although proving rigorous bounds on mixing times of Lindbladians is notoriously difficult.
We expect that Gibbs states and thermal phases can generally be described by the Lindbladian definition of matter, but it remains an open question to establish this rigorously.

\ \newline
\noindent \textbf{Beyond Lindbladian Phases.}
Beyond the definition of Lindbladian phases, there exist other definitions of \textit{phase of matter} which may be physically relevant and be worth investigating \cite{altland2021symmetry, de2022symmetry, molignini2023topological, ruiz2022matrix, piroli2021quantum, liu2023dissipative}.
    Currently it is unclear how exactly these different definition relate to each other.
    We speculate that any definition of phase characterised by rapid mixing between states under a geometrically local evolution should satisfy results similar to the ones given here.
    Furthermore, is would be interesting to see if the Lindbladian definition of phase can be proven to be the same/different from the steady state definition of phase introduced in \cite{rakovszky2023defining}.
    Indeed, we expect that at the very least, one can learn local patches of phases defined in \cite{rakovszky2023defining}.

        \ \newline
    \noindent \textbf{Learning Phase Boundaries.}
    The results in this work are concerned with learning a known phase.
However, in general, determining where the boundaries of a phase diagram are is not a trivial task.
For the Hamiltonians of \cite{Cubitt_Lucia_Michalakis_Perez-Garcia_2015, Bausch_Cubitt_Lucia_Perez-Garcia_Wolf_2017, Bausch_Cubitt_Watson_2021, watson2021complexity}, the families of Hamiltonians are parameterised by a single parameter either $\varphi\in [0,1]$ or $\varphi\in [0,1] \cap \mathbb{Q}$, but determining the position of the boundary is at least as hard as $\textsf{QMA}_{\textsf{EXP}}$ (in the case of a finite-sized system) or is undecidable (in the thermodynamic limit).
Moreover, common techniques for approximating phase boundaries must fail due to these complexity/computability results \cite{watson2022uncomputably}.

Huang et al. \cite{huang2021provably} show that for a finite-sized system there is a provable sample-efficient advantage. 
They demonstrate one can use Support Vector Machine techniques and samples from the separate phases to efficiently learn the phase boundary with polynomially many samples, \textit{provided there is a local order parameter}.
The case where there no samples are provided is $\textsf{P}^{\textsf{QMA}_{\textsf{EXP}}}$-complete in general \cite{watson2021complexity}.
For the general case of no local order parameter, the sample complexity remains an open question, and for many topological phases is believed to be exponential.

    \ \newline
    \noindent \textbf{Further Learning Techniques.} The learning algorithm in this work generates an estimator for the local observable of interest by an empirical average.
    We conjecture that if one uses the techniques used here to prove smoothness of local observables and concentration of measure in combination with more advanced techniques from machine learning, then one may well be able to greatly improve the sample complexity required in practice.

%\subsection{Related Work}

%The work here is not the first to study mixing times and relate it to the static the idea of phase.
%The study of Glauber dynamics applied to \textit{classical} lattice spin systems has produced many fruitful results.
%\cite{stroock1992logarithmic, martinelli19942, dyer2004mixing} demonstrate the systems mixing in time $O(n\log(n/\epsilon))$ also satisfy ``strong spatial mixing'' equivalent to \textcolor{red}{XXX}, demonstrating they can be learned quickly by the methods presented in this work.
%Since Glauber dynamics only acts on a single lattice site at each time step, morally it can be viewed as having a parallelised mixing in time $O(\log(n/\epsilon))$ (this can be made rigorous with the idea of parallel Glauber dynamics \cite{lee2023parallelising}).

    \section*{Acknowledgements}
	
    The authors gratefully recognise useful discussions with
    \begingroup
    \hypersetup{urlcolor=navyblue}
    \href{https://orcid.org/0000-0003-2990-791X}{David P\'erez-Garc\'ia} on the topic of Lindbladian phases of matter, and for suggesting \cref{Condition:Compatability_Condition} as an added requirement for a reasonable definition of ``phase of matter''.
    We thank \href{https://orcid.org/0000-0002-0335-9508}{Victor Albert} for his useful comments, \href{https://orcid.org/0000-0001-8121-2977}{Andrew Guo} for helpful discussions concerning mixing of Lindbladians, and \href{https://orcid.org/0000-0001-7898-0211}{Cheng-Ju (Jacob) Lin} for discussions about the definition of mixed state phase.
    The authors also appreciate discussions about the definition of phase for open systems with \href{https://orcid.org/0009-0008-7293-5952}{Tibor Rakovszky},  \href{https://orcid.org/0000-0002-7493-7600}{Sarang Gopalakrishnan}, and \href{https://orcid.org/0000-0001-6235-6430}{Curt von Keyserlingk}.

	\endgroup
	
    EO is supported by the European Research Council under grant agreement no. 101001976 (project EQUIPTNT) and by the German Research Foundation DFG via the SFB/Transregio~352.
	CR acknowledges financial support from a Junior Researcher START Fellowship from the DFG cluster of excellence 2111 (Munich Center for Quantum Science and Technology), from the ANR project QTraj (ANR-20-CE40-0024-01) of the French National Research Agency (ANR), as well as from the Humboldt Foundation.
	DSF is supported by France 2030 under the French National Research Agency award number “ANR-22-PNCQ-0002”.
	JDW acknowledges support from the United States Department of Energy, Office of Science, Office of Advanced Scientific Computing Research, Accelerated Research in Quantum Computing program, and also NSF QLCI grant OMA-2120757.

\section*{Author Contribution Statement}
JDW was the primary contributor to this work.
All other authors contributed equally.

\bibliography{References}

\newcommand{\etalchar}[1]{$^{#1}$}
\begin{thebibliography}{RdAGRMPG22}

\bibitem[ABFJ16]{albert2016geometry}
Victor~V Albert, Barry Bradlyn, Martin Fraas, and Liang Jiang.
\newblock Geometry and response of lindbladians.
\newblock {\em Physical Review X}, 6(4):041031, 2016.

\bibitem[AFD21]{altland2021symmetry}
Alexander Altland, Michael Fleischhauer, and Sebastian Diehl.
\newblock Symmetry classes of open fermionic quantum matter.
\newblock {\em Physical Review X}, 11(2):021037, 2021.

\bibitem[Amb14]{ambainis2014physical}
Andris Ambainis.
\newblock On physical problems that are slightly more difficult than qma.
\newblock In {\em 2014 IEEE 29th Conference on Computational Complexity (CCC)},
  pages 32--43. IEEE, 2014.

\bibitem[BCG{\etalchar{+}}23]{bardet2023rapid}
Ivan Bardet, {\'A}ngela Capel, Li~Gao, Angelo Lucia, David
  P{\'e}rez-Garc{\'\i}a, and Cambyse Rouz{\'e}.
\newblock Rapid thermalization of spin chain commuting hamiltonians.
\newblock {\em Physical Review Letters}, 130(6):060401, 2023.

\bibitem[BCGW22]{bravyi2022quantum}
Sergey Bravyi, Anirban Chowdhury, David Gosset, and Pawel Wocjan.
\newblock Quantum hamiltonian complexity in thermal equilibrium.
\newblock {\em Nature Physics}, 18(11):1367--1370, 2022.

\bibitem[BCL{\etalchar{+}}17]{Bausch_Cubitt_Lucia_Perez-Garcia_Wolf_2017}
Johannes Bausch, Toby~S. Cubitt, Angelo Lucia, David Perez-Garcia, and
  Michael~M. Wolf.
\newblock Size-driven quantum phase transitions.
\newblock {\em Proceedings of the National Academy of Sciences},
  115(1):19–23, 2017.

\bibitem[BCW21]{Bausch_Cubitt_Watson_2021}
Johannes Bausch, Toby~S. Cubitt, and James~D. Watson.
\newblock Uncomputability of phase diagrams.
\newblock {\em Nature Communications}, 12(1), 2021.

\bibitem[BGL20]{barr2020quantum}
Ariel Barr, Willem Gispen, and Austen Lamacraft.
\newblock Quantum ground states from reinforcement learning.
\newblock In {\em Mathematical and Scientific Machine Learning}, pages
  635--653. PMLR, 2020.

\bibitem[BK19]{brandao2019finite}
Fernando~GSL Brand{\~a}o and Michael~J Kastoryano.
\newblock Finite correlation length implies efficient preparation of quantum
  thermal states.
\newblock {\em Communications in Mathematical Physics}, 365:1--16, 2019.

\bibitem[BMNS12]{bachmann2012automorphic}
Sven Bachmann, Spyridon Michalakis, Bruno Nachtergaele, and Robert Sims.
\newblock Automorphic equivalence within gapped phases of quantum lattice
  systems.
\newblock {\em Communications in Mathematical Physics}, 309(3):835--871, 2012.

\bibitem[BWP{\etalchar{+}}17]{Biamonte2017}
Jacob Biamonte, Peter Wittek, Nicola Pancotti, Patrick Rebentrost, Nathan
  Wiebe, and Seth Lloyd.
\newblock Quantum machine learning.
\newblock {\em Nature}, 549(7671):195--202, September 2017.

\bibitem[CB23]{coopmans2023sample}
Luuk Coopmans and Marcello Benedetti.
\newblock On the sample complexity of quantum boltzmann machine learning.
\newblock {\em arXiv preprint arXiv:2306.14969}, 2023.

\bibitem[CCK{\etalchar{+}}23]{chen2023thermal}
{Chi-Fang}, {Chen}, Michael~J. {Kastoryano}, Fernando G.~S.~L. {Brand{\~a}o},
  and Andr{\'a}s {Gily{\'e}n}.
\newblock {Quantum Thermal State Preparation}.
\newblock {\em arXiv e-prints}, page arXiv:2303.18224, March 2023.

\bibitem[CGN23]{che2023exponentially}
Yanming Che, Clemens Gneiting, and Franco Nori.
\newblock Exponentially improved efficient machine learning for quantum
  many-body states with provable guarantees.
\newblock {\em arXiv preprint arXiv:2304.04353}, 2023.

\bibitem[CGW10]{chen2010local}
Xie Chen, Zheng-Cheng Gu, and Xiao-Gang Wen.
\newblock Local unitary transformation, long-range quantum entanglement, wave
  function renormalization, and topological order.
\newblock {\em Physical review b}, 82(15):155138, 2010.

\bibitem[CLMPG15]{Cubitt_Lucia_Michalakis_Perez-Garcia_2015}
Toby~S. Cubitt, Angelo Lucia, Spyridon Michalakis, and David Perez-Garcia.
\newblock Stability of local quantum dissipative systems.
\newblock {\em Communications in Mathematical Physics}, 337(3):1275–1315,
  2015.

\bibitem[CM17]{Carrasquilla2017}
Juan Carrasquilla and Roger~G. Melko.
\newblock Machine learning phases of matter.
\newblock {\em Nature Physics}, 13(5):431--434, February 2017.

\bibitem[CPG19]{Coser_Perez-Garcia_2019}
Andrea Coser and David Pérez-García.
\newblock Classification of phases for mixed states via fast dissipative
  evolution.
\newblock {\em Quantum}, 3:174, 2019.

\bibitem[CPGSV21]{cirac2021matrix}
J~Ignacio Cirac, David Perez-Garcia, Norbert Schuch, and Frank Verstraete.
\newblock Matrix product states and projected entangled pair states: Concepts,
  symmetries, theorems.
\newblock {\em Reviews of Modern Physics}, 93(4):045003, 2021.

\bibitem[CRF20]{capel2020modified}
Angela Capel, Cambyse Rouz{\'e}, and Daniel~Stilck Fran{\c{c}}a.
\newblock The modified logarithmic sobolev inequality for quantum spin systems:
  classical and commuting nearest neighbour interactions.
\newblock {\em arXiv preprint arXiv:2009.11817}, 2020.

\bibitem[Cub23]{cubitt2023dissipative}
Toby~S Cubitt.
\newblock Dissipative ground state preparation and the dissipative quantum
  eigensolver.
\newblock {\em arXiv preprint arXiv:2303.11962}, 2023.

\bibitem[dGTS22]{de2022symmetry}
Caroline de~Groot, Alex Turzillo, and Norbert Schuch.
\newblock Symmetry protected topological order in open quantum systems.
\newblock {\em Quantum}, 6:856, 2022.

\bibitem[DPZ19]{Dong_Pollmann_Zhang_2019}
Xiao-Yu Dong, Frank Pollmann, and Xue-Feng Zhang.
\newblock Machine learning of quantum phase transitions.
\newblock {\em Phys. Rev. B}, 99:121104, Mar 2019.

\bibitem[FG18]{fischer2018simple}
Manuela Fischer and Mohsen Ghaffari.
\newblock A simple parallel and distributed sampling technique: Local glauber
  dynamics.
\newblock {\em arXiv preprint arXiv:1802.06676}, 2018.

\bibitem[GD17]{gao2017efficient}
Xun Gao and Lu-Ming Duan.
\newblock Efficient representation of quantum many-body states with deep neural
  networks.
\newblock {\em Nature communications}, 8(1):662, 2017.

\bibitem[GKW16]{gubernatis2016quantum}
James Gubernatis, Naoki Kawashima, and Philipp Werner.
\newblock {\em Quantum Monte Carlo Methods}.
\newblock Cambridge University Press, 2016.

\bibitem[GLTG21]{guo2021clustering}
Andrew~Y Guo, Simon Lieu, Minh~C Tran, and Alexey~V Gorshkov.
\newblock Clustering of steady-state correlations in open systems with
  long-range interactions.
\newblock {\em arXiv preprint arXiv:2110.15368}, 2021.

\bibitem[GPY19]{gharibian2019oracle}
Sevag Gharibian, Stephen Piddock, and Justin Yirka.
\newblock Oracle complexity classes and local measurements on physical
  hamiltonians.
\newblock {\em arXiv preprint arXiv:1909.05981}, 2019.

\bibitem[Has11]{hastings2011topological}
Matthew~B Hastings.
\newblock Topological order at nonzero temperature.
\newblock {\em Physical review letters}, 107(21):210501, 2011.

\bibitem[HKP20]{Huang2020}
H.-Y. Huang, R.~Kueng, and J.~Preskill.
\newblock Predicting many properties of a quantum system from very few
  measurements.
\newblock {\em Nat. Phys.}, 16(10):1050--1057, 2020.

\bibitem[HKT{\etalchar{+}}22]{huang2021provably}
Hsin-Yuan Huang, Richard Kueng, Giacomo Torlai, Victor~V. Albert, and John
  Preskill.
\newblock Provably efficient machine learning for quantum many-body problems.
\newblock {\em Science}, 377(6613), 2022.

\bibitem[KSV02]{kitaev2002classical}
Alexei~Yu Kitaev, Alexander Shen, and Mikhail~N Vyalyi.
\newblock {\em Classical and quantum computation}.
\newblock Number~47. American Mathematical Soc., 2002.

\bibitem[KT13]{kastoryano2013quantum}
Michael~J Kastoryano and Kristan Temme.
\newblock Quantum logarithmic {S}obolev inequalities and rapid mixing.
\newblock {\em Journal of Mathematical Physics}, 54(5):052202, 2013.

\bibitem[LHT{\etalchar{+}}23]{Lewis_Huang_Preskill_2023}
Laura {Lewis}, Hsin-Yuan {Huang}, Viet~T. {Tran}, Sebastian {Lehner}, Richard
  {Kueng}, and John {Preskill}.
\newblock {Improved machine learning algorithm for predicting ground state
  properties}.
\newblock {\em arXiv e-prints}, page arXiv:2301.13169, January 2023.

\bibitem[LL23]{liu2023dissipative}
Yu-Jie Liu and Simon Lieu.
\newblock Dissipative phase transitions and passive error correction.
\newblock {\em arXiv preprint arXiv:2307.09512}, 2023.

\bibitem[LP17]{levin2017markov}
David~A Levin and Yuval Peres.
\newblock {\em Markov chains and mixing times}, volume 107.
\newblock American Mathematical Soc., 2017.

\bibitem[MC23]{molignini2023topological}
Paolo Molignini and Nigel~R Cooper.
\newblock Topological phase transitions at finite temperature.
\newblock {\em Physical Review Research}, 5(2):023004, 2023.

\bibitem[MOS94]{martinelli19942}
Fabio Martinelli, Enzo Olivieri, and Roberto~H Schonmann.
\newblock For 2-d lattice spin systems weak mixing implies strong mixing.
\newblock {\em Communications in Mathematical Physics}, 165(1):33--47, 1994.

\bibitem[MSW03]{martinelli2003ising}
Fabio Martinelli, Alistair Sinclair, and Dror Weitz.
\newblock The ising model on trees: Boundary conditions and mixing time.
\newblock In {\em 44th Annual IEEE Symposium on Foundations of Computer
  Science, 2003. Proceedings.}, pages 628--639. IEEE, 2003.

\bibitem[NVZ11]{nachtergaele2011lieb}
Bruno Nachtergaele, Anna Vershynina, and Valentin~A Zagrebnov.
\newblock Lieb-robinson bounds and existence of the thermodynamic limit for a
  class of irreversible quantum dynamics.
\newblock {\em AMS Contemporary Mathematics}, 552:161--175, 2011.

\bibitem[NYN21]{nomura2021purifying}
Yusuke Nomura, Nobuyuki Yoshioka, and Franco Nori.
\newblock Purifying deep boltzmann machines for thermal quantum states.
\newblock {\em Physical review letters}, 127(6):060601, 2021.

\bibitem[ORSW23]{Onorati_Rouze_Stilch_Franca_Watson_2023}
Emilio {Onorati}, Cambyse {Rouz{\'e}}, Daniel {Stilck Fran{\c{c}}a}, and
  James~D. {Watson}.
\newblock {Efficient learning of ground \& thermal states within phases of
  matter}.
\newblock {\em arXiv e-prints}, page arXiv:2301.12946, January 2023.

\bibitem[PK20]{park2020geometry}
Chae-Yeun Park and Michael~J Kastoryano.
\newblock Geometry of learning neural quantum states.
\newblock {\em Physical Review Research}, 2(2):023232, 2020.

\bibitem[PMS{\etalchar{+}}14]{peruzzo2014variational}
Alberto Peruzzo, Jarrod McClean, Peter Shadbolt, Man-Hong Yung, Xiao-Qi Zhou,
  Peter~J Love, Al{\'a}n Aspuru-Guzik, and Jeremy~L O’brien.
\newblock A variational eigenvalue solver on a photonic quantum processor.
\newblock {\em Nature communications}, 5(1):4213, 2014.

\bibitem[Pou10]{poulin2010lieb}
David Poulin.
\newblock Lieb-{R}obinson bound and locality for general {M}arkovian quantum
  dynamics.
\newblock {\em Physical review letters}, 104(19):190401, 2010.

\bibitem[PSC21]{piroli2021quantum}
Lorenzo Piroli, Georgios Styliaris, and J~Ignacio Cirac.
\newblock Quantum circuits assisted by local operations and classical
  communication: Transformations and phases of matter.
\newblock {\em Physical Review Letters}, 127(22):220503, 2021.

\bibitem[RdAGRMPG22]{ruiz2022matrix}
Alberto Ruiz-de Alarc{\'o}n, Jos{\'e} Garre-Rubio, Andr{\'a}s Moln{\'a}r, and
  David P{\'e}rez-Garc{\'\i}a.
\newblock Matrix product operator algebras ii: phases of matter for 1d mixed
  states.
\newblock {\em arXiv preprint arXiv:2204.06295}, 2022.

\bibitem[RGvK23]{rakovszky2023defining}
Tibor Rakovszky, Sarang Gopalakrishnan, and Curt von Keyserlingk.
\newblock Defining stable phases of open quantum systems.
\newblock {\em arXiv preprint arXiv:2308.15495}, 2023.

\bibitem[RKT{\etalchar{+}}19]{rem2019identifying}
Benno~S Rem, Niklas K{\"a}ming, Matthias Tarnowski, Luca Asteria, Nick
  Fl{\"a}schner, Christoph Becker, Klaus Sengstock, and Christof Weitenberg.
\newblock Identifying quantum phase transitions using artificial neural
  networks on experimental data.
\newblock {\em Nature Physics}, 15(9):917--920, 2019.

\bibitem[RNS19]{rodriguez2019identifying}
Joaquin~F Rodriguez-Nieva and Mathias~S Scheurer.
\newblock Identifying topological order through unsupervised machine learning.
\newblock {\em Nature Physics}, 15(8):790--795, 2019.

\bibitem[Vid08]{vidal2008class}
Guifr{\'e} Vidal.
\newblock Class of quantum many-body states that can be efficiently simulated.
\newblock {\em Physical review letters}, 101(11):110501, 2008.

\bibitem[VWIC09]{verstraete2009quantum}
Frank Verstraete, Michael~M Wolf, and J~Ignacio~Cirac.
\newblock Quantum computation and quantum-state engineering driven by
  dissipation.
\newblock {\em Nature physics}, 5(9):633--636, 2009.

\bibitem[WB21]{watson2021complexity}
James~D Watson and Johannes Bausch.
\newblock The complexity of approximating critical points of quantum phase
  transitions.
\newblock {\em arXiv preprint arXiv:2105.13350}, 2021.

\bibitem[Whi92]{white1992density}
Steven~R White.
\newblock Density matrix formulation for quantum renormalization groups.
\newblock {\em Physical review letters}, 69(19):2863, 1992.

\bibitem[WOC22]{watson2022uncomputably}
James~D Watson, Emilio Onorati, and Toby~S Cubitt.
\newblock Uncomputably complex renormalisation group flows.
\newblock {\em Nature Communications}, 13(1):7618, 2022.

\end{thebibliography}
\bibliographystyle{alpha}
\clearpage

\appendix

\section*{Supplementary Material}

\section{Notation and definition of Lindbladian superoperators}

Given a finite dimensional Hilbert space $\mathcal{H}$,
% and $\mathcal{H}$ be finite dimensional Hilbert spaces.
we denote by $\mathcal{B}(\mathcal{H})$ 
% and $\mathcal{B}(\mathcal{H})$ denote 
{the algebra of bounded operators on $\mathcal{H}$}, whereas $\mathcal{B}_{\operatorname{sa}}(\cH)$ denotes the subspace of self-adjoint operators.
% and $\mathcal{H}$, respectively. 
We denote by $\mathcal{D}(\mathcal{H})$ the set of positive operators on $\mathcal{H}$ of unit trace, and by $\mathcal{D}_{+}(\mathcal{H})$ the subset of positive, full-rank operators on $\mathcal{H}$. 
Schatten norms are denoted by $\|.\|_p$ for $p\ge 1$. The identity matrix in $\mathcal{B}(\mathcal{H})$ is denoted by $I$. 
In this work, systems on a $D$-dimensional lattice $\Lambda=[-L,L]^D$ where at each site $x$ of the lattice we will associate a $d$-dimension Hilbert space $\mathcal{H}_x$.
We denote the subset of lattice points within a ball of radius $r$, centred on $u$ as $b_u(r)$.
The algebra of observables supported on $\Lambda$ is defined as $\mathcal{A}_\Lambda = \bigotimes_{x\in \Lambda}\mathcal{B}(\mathcal{H}_x)$.
A linear map $\mathcal{T}: \mathcal{A} \rightarrow \mathcal{A}$ is called a superoperator.

\bigskip

The dynamics of a Markovian, dissipative system is generated by a superoperator known as a \emph{Lindbladian}.
In this work, we consider a family of local, time-independent Lindbladians acting over the $D$-dimensional lattice $\Lambda=[-L,L]^D$.
Lindbladians have the form:
\begin{align}
    \mathcal{L}(\rho) = i[\rho,H] + \sum L_j \rho L_j^* - \frac{1}{2}\sum_j \{L_j^*L_j,\rho \}
\end{align}
where $H$ is a Hermitian matrix, $\{L_j\}_j$ a set of matrices called the Lindblad operators.
$[A, B]$ and $\{A, B\}$ denote 
the commutator and the anticommutator respectively of operators $A,B$.
The time evolution operator is then given by $e^{t\L}$.
The fact that the Lindbladian is local means it can be written as a sum of local terms:
\begin{align*}
   \L(x)= \sum_{j}  \L_{j}(x_j),
\end{align*}
where we have allow the terms to depend on some parameterisation, $x\in \mathbb{R}^m$, and where each $x_j\in \mathbb{R}^\ell$, and $\ell =O(1)$. 
For the rest of this work we will assume that the Lindbladian is linearly parameterised, $\L_j(x_j) = x_jE_j$, for an operator $E_j$.
However, the result apply equally for Lindbladians with bounded derivatives.
We will often characterise a phase as a subset of parameter space $\Phi\subseteq \mathbb{R}^m$.

Each $\L_{j}(x_j)$ supported on a ball $A_j$ around site $j\in\Lambda $ of radius $r_0$.
Thus the interactions have finite range.
We will also use the notation $\L_{u,r}$ to denote a term with support entirely within $b_u(r)$.
Given a vector $x \in \mathbb{R}^m$, a subsection of the lattice $S\subseteq \Lambda $ and a local Lindbladian $\L(x)$, then $x_{|S}$ refers to the restriction of $x$ to only those parameters which describe local terms in $\L_j(x_j)$ which have support on $S$.

Consider two families of local Lindbladians, $\L, \L' $ with local terms $\L_j, \L_j'$ respectively.
For some $A\subseteq \Lambda$, we will also use the notation $\L^A(x) = \sum_{j:\supp(\L_j)\subseteq A } \L_j(x_j) + \sum_{j:\supp(\L_j)\not\subseteq A } \L'_j(x_j)$.
That is, the Lindbladian has terms from $\L$ within $A$, and terms from $\L'$ outside of $A$.

\begin{definition}[Lindbladian Steady State]
    We say that $\L$ has a unique steady state $\rho_\infty$ if, for all density matrices $\rho$, $\lim_{t\rightarrow \infty}e^{t\L}(\rho)=\rho_\infty $.
    For parameterised sets of Lindbladians $\{\L(x)\}_{x}$, each with a unique steady state, we denote the corresponding steady state by $\rho_\infty(x)$.
    %, 
\end{definition}

Any given Lindbladian forms the generator of a one parameter semigroup which represents the evolution of the system through time.
We denote this $T_t = e^{t\L}$, and its dual is denoted $T_t^*= e^{t\L^*}$, such that for a state $\rho$ and an observable $A$, then
$\tr[AT_t(\rho)] = \tr[\rho T^*_t(A)]$.

We can then define the relevant completely bounded norms for these superoperators as
\begin{align*}
    \norm{T}_{1\rightarrow 1,cb} = \sup_n\norm{T\otimes \mathds{1}_n}_{1\rightarrow 1} = \sup_n \sup_{X\in A_\Lambda\otimes M_n, X\neq 0} \frac{\norm{T\otimes \mathds{1}_n(X)}_{1}}{\norm{X}_1}.
\end{align*}
and similarly
\begin{align*}
     \norm{T}_{\infty \rightarrow \infty,cb} = \sup_n\norm{T\otimes \mathds{1}_n}_{\infty\rightarrow \infty} = \sup_n \sup_{X\in A_\Lambda\otimes M_n, X\neq 0} \frac{\norm{T\otimes \mathds{1}_n(X)}_{\infty}}{\norm{X}_\infty}.
\end{align*}
We note that $\norm{T}_{1\rightarrow 1,cb} = \norm{T^*}_{\infty\rightarrow \infty,cb}$.
 We define the strength of the Lindbladian as:
\begin{align*}
    J = \sup_{u,r}\norm{\L_{u,r}}_{1\rightarrow 1, cb}, \quad f(r) = \sup_u \frac{\norm{\L_{u,r}}_{1\rightarrow 1, cb}}{J}.
\end{align*}
%We will also assume that $\norm{\partial_{x_{j,r}}\L_{j}(x_j)}_{1\rightarrow 1, cb}\leq 1$.
For the remainder of this work, we will assume that we are working with finite range interactions, and so $f(r)$ is compactly supported. 
However, these results can be adapted to exponentially decaying and sufficiently fast decaying polynomial decays. 
When talking about a specific observable $O$, we will assume it is supported on a ball of diameter at most $k_0$ around site $i$ in the lattice.

Finally, we need to restrict the family of Lindbladians we are restricting to, realising that we need to be able to talk about Lindbladians at different systems sizes:
\begin{definition}[Uniform Families of Lindbladians] \label{Def:Uniform_Families}
A uniform family of Lindbladians $\L$ with strength $(J,f) $ is given by the following:
\begin{itemize}
    \item infinite Lindbladian: a Lindbladian M on all of $\mathbb{Z}^D$ with strength $(J,f)$;
    \item boundary conditions: we assume either open or periodic boundary conditions.
\end{itemize}
From here onward, we impose \cref{Def:Uniform_Families} on all Lindbladians which act on the primary system, but not necessarily on the auxiliary system.
A Lindbladian acting on a finite lattice is then defined by the restriction of the infinite Lindbladian to that subsection of the lattice, plus the relevant boundary conditions.
    
\end{definition}
\noindent An important subset of uniform families of Lindbladians are those with translationally invariant terms, which naturally define a Lindbladian on an infinite lattice.

\subsection{Toolbox}

A key tool we will utilise for our efficient learning results are Lieb-Robinson bounds, which put a ``speed limit'' on the propagation of information throughout a system having geometrically local interactions.
We will import many of the techniques and lemmas from \cite{Cubitt_Lucia_Michalakis_Perez-Garcia_2015} to prove our result.

\iffalse
\textcolor{red}{Check this.}
\begin{assumption}[Lieb-Robinson Condition]
Let $\L=\sum_{j} \L_{u,r}$ be a local Lindbladian.
Then there exists a positive constant $ v$ such that:
\begin{align*}
    \sup_{x\in \Lambda} \sum_{u\in \Lambda}\sum_{r\geq dist(u,x)} \norm{\L_{u,r}}_{1\rightarrow 1,cb}|b_u(r)|\nu(r) \leq \frac{v}{2}<\infty.
 \end{align*}
 \textcolor{red}{Define $b(r)$ and make sense of this entire condition.} Furthermore, 
 \begin{align*}
     \nu (r) = \nu^{-1} (r) =  \begin{cases} 1 & r\leq r_0 \\ 
                   0 & r>r_0  \end{cases}.
 \end{align*}
 \textcolor{red}{May need to talk about the inverse here.}
\end{assumption}
\fi

\begin{assumption}[Lieb-Robinson Condition] \label{Assumption:Lieb-Robinson_Condition}
Let $\L=\sum_{u,r} \L_{u,r}$ be a local Lindbladian, where $\L_{u,r}$ is a Lindbladian term with support entirely contained within a ball of radius $r$ centred at $u$, denoted $b_u(r)$.
We say $\mathcal{L}$ satisfies the LR condition if and only if there exist positive constants $\mu, v$ such that:
\begin{align*}
    \sup_{x\in \Lambda} \sum_{u\in \Lambda}\sum_{r\geq dist(u,x)} \norm{\L_{u,r}}_{1\rightarrow 1,cb}|b_u(r)|\nu_{\mu}(r) \leq \frac{v}{2}<\infty.
 \end{align*}
 Furthermore, 
 \begin{align*}
     \nu_\mu (r) = e^{\mu r},
 \end{align*}
 and $\nu_\mu^{-1}(r) =1/\nu_\mu(r) $.
\end{assumption}

\noindent This is satisfied if $\L$ has finite ranged or exponentially decaying interactions \cite[Remark 5.2]{Cubitt_Lucia_Michalakis_Perez-Garcia_2015}.
This assumption allows us to put limits on the propagation of correlations and other effects as the system evolves.
Importantly, using the Lieb-Robinson bounds for dissipative systems from \cite{poulin2010lieb, nachtergaele2011lieb}, the following lemma can be derived:
\begin{lemma}[Localising the Evolution, Lemma 5.5 of \cite{Cubitt_Lucia_Michalakis_Perez-Garcia_2015}] \label{Lemma:Localised_Evolution}
Let $O_A$ be a local observable supported on $A\subset \Lambda$, and let $A(r)$ denote a ball of radius $r$ around $A$. 
Denote $O_A(t) = T^{*}_t(O_A)$ its evolution under a local Lindbladian $\L$ with strength $(J,f)$.
Given $r>0$, denoted by $O_A(r;t)$ its evolution under the Lindbladian $\L^{A(r)}$ with strength $(J,f)$, which is identical to $\L$ within $A(r)$, but may differ outside.
Then, using the notation from \cref{Assumption:Lieb-Robinson_Condition}, the following bound holds:
\begin{align*}
    \norm{O_A(t) - O_A(t;r) }\leq \norm{O_A}|A|J \frac{e^{vt}-1-vt}{v} \nu_{\beta}^{-1}(r),
\end{align*}
for a constant $\beta$.
\end{lemma}

\section{Learning Lindbladian Phases of Matter}
\label{Sec:Learning_Lindbladian_Phases}

In this section we show that phases satisfying the Lindbladian definition of phase, in addition to the compatibility condition (\cref{Condition:Compatability_Condition}), can be learned efficiently.
One of the key components to this is to show that by perturbing the Lindbladian, the properties of the states along the trajectory of the Lindbladian do no vary too much.
There are multiple results showing properties similar to this under various assumptions, e.g. \cite{albert2016geometry, guo2021clustering, liu2023dissipative}.
However, for our purposes, we follow a similar path to the proof of \cite{Cubitt_Lucia_Michalakis_Perez-Garcia_2015}.

The first step of the proof is to prove that the local rapid mixing condition holds:
\begin{condition}[Local Rapid Mixing]
\label{Assumption:Local_Rapid_Mixing}
    \ Let $\L$ be a Lindbladian with steady state $\rho_\infty$, and let $\rho_0$ be some point in the same phase as $\rho_\infty$ (e.g. the reference point of the phase, $\refrho$), such that $\rho_0$ is the steady state of the Lindbladian $\L'$.
 Let $O_A$ be a geometrically local observable supported on $A\subset \Lambda$.
Then:
    \begin{align*}
       | \tr[O_A(T_t(\rho_0) - \rho_\infty)]| \leq \poly(|A|) e^{-\gamma t}
       %2 \left(  \frac{J}{v}|A| + c|A|^\delta  \right)p(t(v+\gamma)/\beta) e^{-\gamma t}
    \end{align*}
\end{condition}

We first prove a LTQO-type condition, which states that local observables on $A$ are in some sense indistinguishable between steady states of different sizes.
\begin{lemma}[LTQO-Type Condition] \label{Lemma:LTQO_Condition}
    Let $\L$ be a uniform family of Lindbladians which for any $A\subset \Lambda $, the Lindbladian has local terms of $\L$ within $A$ and terms of some other Lindbladian outside of $A$.
  %  Assume $\L^{A}$ has a unique steady state, and assume its steady state is in the same phase as $\L$.
    Let $O_A$ be an observable supported on $A\subset \Lambda$ and let $\rho^{A(s)}_\infty$ be the steady state of $\L^{A(s)}$ and let $\rho_\infty$ be the unique steady state of $\L$.
    Assume $\rho^{A(s)}_\infty$ and $\rhofixed$ are in the same phase.
    Then:
    \begin{align*}
        \sup_{O_A} |\tr[O_A(\rho_\infty - \rho^{A(s)}_\infty)]| \leq \norm{O_A}\bigg( \frac{J}{v}|A| + c|A|^\kappa \bigg)\Delta_0(s)
    \end{align*}
    where 
     \begin{align*}
        \Delta_0(s) = \bigg( \frac{|A(s)|}{|A|} \bigg)^{\kappa v/(v+\gamma)}\nu_{\beta'}^{-1}(s)
    \end{align*}
    and $\beta' = \beta\gamma/(v+\gamma)$.
\end{lemma}
%\begin{lemma}[LTQO-Type Condition] \label{Lemma:LTQO_Condition}
%    Let $\L$ be a uniform family of Lindbladians, and suppose for any $A\subset \Lambda $, the restricted Lindbladian $\L^{A}$ has a unique steady state.
%    Let $O_A$ be an observable supported on $A\subset \Lambda$ and let $\rho^{A(s)}_\infty$ be the steady state of $\L^{A(s)}$ and let $\rho_\infty$ be the unique steady state of $\L$.
%    Then:
%    \begin{align*}
%        \sup_{O_A} |\tr[O_A(\rho_\infty - \rho^{A(s)}_\infty)]| \leq \norm{O_A}\bigg( \frac{J}{v}|A| + c|A|^\delta \bigg)\Delta_0(s)
%    \end{align*}
%    where 
%     \begin{align*}
%        \Delta_0(s) = \bigg( \frac{|A(s)|}{|A|} \bigg)^{\delta v/(v+\gamma)}\nu_{\beta'}^{-1}(s)
%    \end{align*}
%    and $\beta' = \beta\gamma/(v+\gamma)$.
%\end{lemma}
\begin{proof}
    Applying the triangle inequality gives:
    \begin{align*}
        |\tr[O_A(\rho_\infty - \rho^{A(s)}_\infty)]| 
        &\leq |\tr[O_A(\rho_\infty - T_t^{A(s)}(\rho_\infty))]|  + |\tr[O_A(T_t^{A(s)}(\rho_\infty) - \rho^{A(s)}_\infty)]| \\
        &= |\tr[O_A(T_t(\rho_\infty) - T_t^{A(s)}(\rho_\infty))]|  + |\tr[O_A(T_t^{A(s)}(\rho_\infty) - \rho^{A(s)}_\infty)]|
    \end{align*}
    The first term can be bounded by applying the Lieb-Robinson bounds summed up in \cref{Lemma:Localised_Evolution} giving:
    \begin{align*}
         |\tr[O_A(T_t(\rho_\infty) - T_t^{A(s)}(\rho_\infty))]| &=\tr[\rho_\infty(T_t^*(O_A) - T_t^{*A(s)}(O_A))] \\
         &\leq \norm{\rho_\infty}_1 \norm{ T_t^*(O_A) - T_t^{*A(s)}(O_A)) } \\
         &\leq \norm{O_A}|A|\frac{J}{v}e^{vt}\nu_\beta^{-1}(s)
    \end{align*}

  \noindent   The second term is bounded by \cref{Condition:Compatability_Condition}:
    \begin{align*}
        |\tr[O_A(T_t^{A(s)}(\rho_\infty) - \rho^{A(s)}_\infty)]| \leq \norm{O_A} |A(s)|^\kappa e^{-\gamma t}.
    \end{align*}
    Putting these together:
    \begin{align*}
         |\tr[O_A(\rho_\infty - \rho^{A(s)}_\infty)]| 
        &\leq \norm{O_A}|A|\frac{J}{v}e^{vt}\nu_\beta^{-1}(s) + \norm{O_A} |A(s)|^\kappa e^{-\gamma t} \\
        &= \norm{O_A}|A|\frac{J}{v}e^{vt}\nu_\beta^{-1}(s) + \norm{O_A} |A|^\kappa p(s)e^{-\gamma t},
    \end{align*}
    where in the last step we have simply taken the $s$ dependence out of $A(s)$ to make it more explicit, $p(s) = (|A(s)|/A)^\kappa$.

    We then choose $t=t(s)$ such that $e^{vt}\nu_\beta^{-1}(s)=p(s)e^{-\gamma t}$,
    \begin{align*}
        t(s) = \log( \nu_\beta(s)p(s) )^{1/(v+\gamma)}.
    \end{align*}
    This gives
        $e^{-\gamma t}  =  \nu_{\beta'}(s)p(s)^{\gamma/(v+\gamma)}$
    where $\beta' = \beta\gamma/(v+\gamma)$, which gives
    \begin{align*}
        \Delta_0(s) = \bigg( \frac{|A(s)|}{|A|} \bigg)^{\kappa v/(v+\gamma)}\nu_{\beta'}^{-1}(s)
    \end{align*}

\end{proof}

By slightly modifying the proof of Proposition 6.6 in \cite{Cubitt_Lucia_Michalakis_Perez-Garcia_2015}, we show the above lemma implies that when moving between states within a phase, we expect even faster mixing of the local observables compared to the state.

\begin{lemma}[Rapid Mixing of Observables within a Phase] \label{Lemma:Observable_Mixing}
 Let $\L$ be a Lindbladian with steady state $\rho_\infty$, and let $\rho_0$ be some point in the same phase as $\rho_\infty$ (e.g. the reference point of the phase, $\refrho$), such that $\rho_0$ is the steady state of the Lindbladian $\L'$.
 Let $\L^{A(s)}$ be a Lindbladian with local terms corresponding to $\L$ within $A(s)$, and $L'$ outside of $A(s)$.
 Let $O_A$ be a geometrically local observable supported on $A\subset \Lambda$.
Then:
    \begin{align*}
       | \tr[O_A(T_t(\rho_0) - \rho_\infty)]| \leq 2 \norm{O_A}\left(  \frac{J}{v}|A| + c|A|^\kappa  \right)p(t(v+\gamma)/\beta) e^{-\gamma t},
    \end{align*}
    where $p$ is a polynomial.

\end{lemma}
\begin{proof}
Denote $s_0$ to be the minimum $s\geq 0$ such that $A(s)=\Lambda$.

    Applying the triangle inequality gives:
    \begin{align*}
        | \tr[O_A(T_t(\rho_0) - \rho_\infty)]| &\leq | \tr[O_A(T_t(\rho_0) - T_t^{A(s)}(\rho_0))]| + | \tr[O_A(T_t^{A(s)}(\rho_0) - T^{A(s)}_\infty(\rho_0))]| \\ 
        &+ | \tr[O_A(T^{A(s)}_\infty(\rho_0) - \rho_\infty)]|.
    \end{align*}

   \noindent  The first term on the left-hand side is bounded by \cref{Lemma:Localised_Evolution}.
    \begin{align*}
        \tr[O_A(T_t(\rho_0) - T_t^{A(s)}(\rho_0))]| &\leq \norm{\rho_0}_1\norm{ (T_t^* - T_t^{*A(s)})(O_A)   } \\
        &\leq \norm{O_A}|A|\frac{J}{v}e^{-\beta s}.
    \end{align*}
 %   To bound the second term, we realise that if $\rho_0$ and $\rho_\infty$ are in the same phase, then by the transitive property of phase $\rho_0$ and $\rho_\infty^{A(s)}$ must be in the same phase.
  %  We can use \cref{Assumption:LTQO} to get the following: \textcolor{red}{Need to check this second part.}
  Then use the rapid mixing property in the compatibility condition, \cref{Condition:Compatability_Condition}, we have: 
  %\textcolor{red}{Double check this.}
    \begin{align*}
        \tr[O_A(T_t^{A(s)}(\rho_0) - T^{A(s)}_\infty(\rho_0))]| \leq c\norm{O_A}|A|^\kappa p(s) e^{-\gamma t}.
    \end{align*}
    
  \noindent   For the third term we apply \cref{Lemma:LTQO_Condition}, which gives:
    \begin{align*}
        | \tr[O_A(T^{A(s)}_\infty(\rho_0) - \rho_\infty)]| &= | \tr[O_A((\rho_\infty^{A(s)}) - \rho_\infty)]| \\
        &\leq \norm{O_A}\bigg(\frac{J}{v}|A|+c|A|^\kappa \bigg)\Delta_0(s).
    \end{align*}

   \noindent  Putting these all together gives:
    \begin{align*}
         | \tr[O_A(T_t(\rho_0) - \rho_\infty)]| \leq \norm{O_A}|A|\frac{J}{v}e^{-\beta s} +  c\norm{O_A}|A|^\kappa p(s) e^{-\gamma t} +\norm{O_A}\bigg(\frac{J}{v}|A|+c|A|^\kappa \bigg)\Delta_0(s)
    \end{align*}
  \noindent   We now choose $s=s(t) \leq s_0$ to such that this term decays exponentially in $t$.
   In particular, we choose $s(t)=t(v+\gamma)/\beta$, which gives:
    \begin{align*}
        \Delta_0(s(t))=p(t(v+\gamma)/\beta)e^{-\gamma t}.
    \end{align*}
  \noindent   Using this, we have: 
    \begin{align*}
         | \tr[O_A(T_t(\rho_0) - \rho_\infty)]| &\leq \norm{O_A}|A|\frac{J}{v}e^{-\beta s} +  c\norm{O_A}|A|^\kappa p(s) e^{-\gamma t} +\norm{O_A}\bigg(\frac{J}{v}|A|+c|A|^\kappa \bigg)\Delta_0(s) \\
         &\leq 2\norm{O_A} \left(  \frac{J}{v}|A| + c|A|^\kappa  \right)p(t(v+\gamma)/\beta) e^{-\gamma t},
    \end{align*}
    which proves the lemma statement.
    
  \noindent As a remark, we note that when $t\geq \beta/(v+\gamma)s_0$, we can simply bound the quantity of interest as:
    \begin{align*}
        | \tr[O_A(T_t(\rho_0) - \rho_\infty)]| \leq c \norm{O_A} |A|^\kappa p(s_0)e^{-\gamma t}.
    \end{align*}

\end{proof}

\noindent To make the exponential decay of with $t$ more explicit, we write the following corollary:
\begin{corollary}[Rapid Mixing of Local Observables II] \label{Corollary:Rapid_Mixing_II}
     \begin{align*}
       | \tr[O_A(T_t(\rho_0) - \rho_\infty)]| \leq 2 \norm{O_A}W\left(  \frac{J}{v}|A| + c|A|^\kappa  \right) e^{-\gamma' t}
    \end{align*}
    for some $\gamma'>0$ and $W=O(1)$.
\end{corollary}
\begin{proof}
    From \cref{Lemma:Observable_Mixing}, we see that $p(t(v+\gamma)/\beta)$ is a polynomial in $t$.
    We can simply choose a $W,\gamma'$ such that
    \begin{align*}
        p(t(v+\gamma)/\beta)e^{-\gamma t} \leq We^{-\gamma't}.
    \end{align*}
\end{proof}
\noindent Thus we have seen that we have obtained the rapid mixing condition.

\subsection{Learnability with Local Rapid Mixing}

%In this section we assume that the Lindbladian satisfies a local rapid mixing condition. 
%We show in \cref{Sec:Rapid_Mixing} this condition can be proven from other physically motivated assumptions.
%However, \cite{Coser_Perez-Garcia_2019} argue that this should hold for systems mapping under a local Lindbladian in their definition of phase, and as such we take it as an assumption for now.

\noindent We are interested in learning the quantity $f_O(\L,t) = \tr[O\rho(\L,t)]$.
We will ignore the effect of the auxiliary state as it is discretely parameterised, and so learning with respect to a different auxiliary state can be formulated separately for each auxiliary state relative to the reference state of interest.

We wish to compare how this quantity is related between ``nearby'' Lindbladians.
The following lemma shows that small perturbations in the Lindbladian only change the expectation value along the evolution by a small amount which is independent of the evolution time.
\begin{lemma}[Stability of Observables within a Phase, Variant of Theorem 20 of \cite{Cubitt_Lucia_Michalakis_Perez-Garcia_2015}] \label{Lemma:Time_Evolution_Difference}
    Let $\L$ be a Lindbladian, and $\L'=\L+E$ another local Lindbladian, where
    \begin{align*}
        E = \sum_{u,r}E_{u,r}
    \end{align*}
    and where the steady states of both $\L,\L'$ are in the same phase.
    Then:
    \begin{align*}
        |f_O(\L,t) - f_O(\L',t)| \leq c(|A|)\norm{O_A}\norm{E^*}_{\infty\rightarrow \infty, cb},
    \end{align*}
    where $c(|A|) = O(\poly(|A|)$.
\end{lemma}
\begin{proof}
Let $T_t$ and $S_t$ be the evolutions generated by $\L$ and $\L'$ respectively. 
    \begin{align*}
        |f_O(\L,t) - f_O(\L',t)| &= |\tr[O_AT_t(\refrho)] - \tr[O_AS_t(\refrho)]| \\
        &=|\tr\refrho[T_t^*(O_A) - S_t^*(O_A)]|.
    \end{align*}
    Using the fact \cite[Eq. (21)]{Cubitt_Lucia_Michalakis_Perez-Garcia_2015}:
    \begin{align*}
        T_t^*(O_A) - S_t^*(O_A) &= \int_0^t S_{t-s}^* E^*T_s^*(O_A) ds \\
        &= \sum_{u,r} \int_0^t S_{t-s}^* E_{u,r}^*T_s^*(O_A) ds
    \end{align*}
    Thus 
    \begin{align*}
        |f_O(\L,t) - f_O(\L',t)| = \sum_{u,r} \tr\left[ \refrho \int_0^t S_{t-s}^* E_{u,r}^*T_s^*(O_A) ds     \right].
    \end{align*}

    Now choose $d = dist(A,\supp E_{u,r})$. 
    We split the integral up into two separate parts: one finishing at time $t_0$ and the second beginning at $t_0$, where $t_0$ is a parameter we are free to choose.
    We first consider:
    \begin{align*}
        \tr\left[ \refrho \int_0^{t_0} S_{t-s}^* E_{u,r}^*T_s^*(O_A) ds     \right] 
        &\leq \norm{\refrho }_1 \int_0^{t_0} \norm{ S_{t-s}^* E_{u,r}^*T_s^*(O_A)} ds \\
        &\leq \int_0^{t_0} \norm{  E_{u,r}^*(O_A)} ds \\
        &\leq \norm{E}_{1\rightarrow 1,cb} \norm{O_A}|A|\frac{2e^{vt_0}-v^2t^2_0-2vt_0}{v^2 \nu_\mu(d)}.
    \end{align*}
    We now set $t_0=t_0(d)$ such that:  
    \begin{align*}
        \nu_\mu^{-1}(d)\frac{e^{vt_0}-v^2t^2_0-vt_0}{v^2} \leq \nu_{\mu/2}^{-1}(d). 
    \end{align*}
    Our choice for $t_0(d)$, which satisfies this, is $t_0(d)=\frac{\mu}{2}\frac{\log(v^2/2)}{v}d$.
    If $t\leq t_0(d)$, then this integral is bounded and we get the result in the lemma statement.

    For $t\geq t_0(d)$ consider:
    \begin{align}
         \tr\left[ \refrho \int_{t_0(d)}^t S_{t-s}^* E_{u,r}^*T_s^*(O_A) ds     \right] &= \tr\left[ \refrho \int_{t_0(d)}^t S_{t-s}^* E_{u,r}^*(T_s^*(O_A) - \tr[O_A\rho_{\infty}]\mathds{1} ) ds     \right] \label{Eq:Second_Int_Line_1} \\
         &\leq \norm{E^*_{u,r}}_{\infty\rightarrow\infty, cb} \int_{t_0(d)}^t \left(\tr[T_s^*(O_A)\refrho] -\tr[O_A\rho_{\infty}]\mathds{1}\right) ds \nonumber \\
         &\leq  c'(|A|) \norm{E^*_{u,r}}_{\infty\rightarrow\infty, cb}\norm{O_A} \int_{t_0(d)}^\infty e^{-\gamma s} ds \label{Eq:Second_Int_Line_2} \\
         &\leq \norm{E^*_{u,r}}_{\infty\rightarrow\infty, cb} \norm{O_A} c'(|A|) \frac{1}{\gamma}e^{-t_0(d)}. \nonumber
    \end{align}
    
    \noindent where in \cref{Eq:Second_Int_Line_1} we have used that $T^*_\infty(O_A) = \tr[O_A\rhofixed]\mathds{1}$, and that $E^*_{u,r}(\mathds{1}) = 0$. 
    To get to \cref{Eq:Second_Int_Line_2} we have used \cref{Assumption:Local_Rapid_Mixing} and where $c'(|A|)=O(\poly(|A|))$.

    We denote $h(d) = e^{-\mu d/2}+ \frac{1}{\gamma}e^{-\gamma t_0(d)}$ which is exponentially decay in $d$.
    Combining the $t\leq t_0$ and $t\geq t_0$ parts gives:
    \begin{align} \label{Eq:Difference_Lindbladian_Single_Term}
         \tr\left[ \refrho \int_0^{t} S_{t-s}^* E_{u,r}^*T_s^*(O_A) ds     \right] \leq \norm{E_{u,r}}_{1\rightarrow 1,cb} \norm{O_A}c'(|A|)h(d),
    \end{align}
     Thus 
    \begin{align} \label{Eq:Sum_of_Differences}
        |f_O(\L,t) - f_O(\L',t)| = \norm{O_A}c'(|A|)\norm{E}_{1\rightarrow 1,cb}\sum_{u,r}   h(d_{u,r}).
    \end{align}
    $h(d)$ decreases exponentially with distance, but there are only a polynomial number terms in the sum, hence the entire sum can be bounded by a constant. 
    Hence we see that:
    \begin{align}\label{Eq:Bound_Different_Lindbladians}
        |f_O(\L,t) - f_O(\L',t)| \leq  c(|A|)\norm{E}_{1\rightarrow 1,cb} \norm{O_A},
    \end{align}
    for some $c(|A|)=\poly(|A|)$ (we have absorbed the sum $\sum_{u,r}   h(d_{u,r})$ and $c'(|A|)$ into the definition of $c(|A|)$).
    
\end{proof}

\noindent Changing the proof slightly, we get the corollary:
\begin{corollary}[Localising Expectation Values] \label{Corollary:Localised_Lindbladian_All_Times}
%    Let $x|_{S(r)}$ denote any vector which has the same elements as $x$ for elements which parameterise terms in $S(r)$, and be different for elements outside of $S(r)$.
    Let $\L(x) = \L + \sum \L_j(x_j)$ be a continuously parameterised family of Lindbladians, where each $\L_j(x_j)$ is a local term supported on $O(1)$ qudits and $x_j\in [-1,1]^\ell$ for a constant $\ell$.
    Let $(x,t), (x|_{S(r)},t)$ correspond to two states in the same phase, %where $x|_{S(r)}$ and $x$ 
    being equal on all parameters in the region $S(r)$, but potentially different outside of $S(r)$.
    Then:
    \begin{align*}
        |f_O(\L(x),t) - f_O(\L(x|_{S(r)}),t)| \leq C_1 e^{-r/2\xi}\norm{O_A},
    \end{align*}
    for an $O(1)$ constant $C_1$.
\end{corollary}
\begin{proof}
    We follow the proof up to \cref{Eq:Sum_of_Differences}, and then write:
      \begin{align*}
        |f_O(\L,t) - f_O(\L',t)| &\leq \norm{O_A}|A|\max_j\norm{\L_j}_{1\rightarrow 1,cb}\sum_{j\in S^c(r)} h(r_j) 
    \end{align*}
    We shift the centre of the lattice to be at the centre of the support of $O$ such that:
     \begin{align*}
        |f_O(\L,t) - f_O(\L',t)| &\leq \norm{O_A}|A|\max_j\norm{\L_j}_{1\rightarrow 1,cb}\sum_{|\ell|>r} h(|\ell|) 
    \end{align*}
    We bound $h(r)\leq c'e^{-r/\xi}$ where $1/\xi=\min\{\gamma, \mu/2 \}$. 
    Summing over these contributions:
     \begin{align*}
        |f_O(\L,t) - f_O(\L',t)| &\leq c'\norm{O_A}|A|J\sum_{|\ell|>r} e^{-|\ell|/\xi} \\
        &= c'\norm{O_A}|A| J \sum_{a>k_0+r/2}  \binom{a+D-1}{D-1} e^{-a/\xi} \\
         &\leq c'\norm{O_A}|A| J D^{D-1}\sum_{a>k_0+r/2} a^{D-1} e^{-a/\xi} \\
         &\leq  c'\norm{O_A}|A| J (D-1)! (2\xi)^{D-1} D^{D-1}\sum_{a>k_0+r/2} a^{D-1} e^{-a/\xi} \\
         &\leq c'\norm{O_A}|A| J (D-1)! (2\xi)^{D-1} D^{D-1} \frac{e^{-(r+1)/\xi}}{1-e^{-1/2\xi}} \\
         &= C_1 e^{-r/2\xi}\norm{O_A},
    \end{align*}
    where the above equation defines $C_1$.
\end{proof}

\subsection{Restricting to Phases Defined by Steady States} \label{Sec:Fixed-Point_Phases}

A natural restriction to consider is the set of states formed by steady states of Lindbladians which we can mixing rapidly between.
Examples of such phases include ground states of local Hamiltonians (which can be written as steady states of local Lindbladians) and also some Gibbs states.
Here we show that we can efficiently learn observables in these phases of matter.
We define them as the following: 
\begin{definition}[Steady State Phase] \label{Def:Fixed_Point_Phase}
    Consider a family of states $\{\rho(\L, \omega, t )\}_{\L,t}$ satisfying \cref{Def:Phase_of_Matter}. 
    We define the steady state phase to be the set of points belonging to this phase which are also steady states of the Lindbladians describing this phase.

    We will further assume that the set of Lindbladians can be continuously parameterised by\footnote{In general we can consider any subset of $\mathbb{R}^m$ such that each parameter varies over an $O(1)$ region.}
    $x\in [-1,1]^m$, $\L(x)$, and all states can be reached by evolution under $e^{t\L}(\rho^*\otimes  \omega^*)$ for some fixed reference states $\rho^*,\omega^*$, hence the states in the phase can be conveniently denoted $\rhofixed(x)  := \lim_{t\rightarrow \infty} e^{t\L(x)}(\rho^*\otimes  \omega^*)$.
\end{definition}
We note an important subtlety here: for two points $\rho(x), \rho(x')$ in the same steady state phase, we do not require $\rhofixed(x)$ to mix rapidly to $\rhofixed(x')$ under $\L(x')$ (although this may also be the case).
We only require that there exists a relevant reference state $\refrho$ which mixes rapidly (in the sense of \cref{Def:Phase_of_Matter}) to $\rhofixed(x)$ under $\L(x)$ and similarly for $\rhofixed(x')$ under $\L(x')$, and that it is possible to rapidly mix back under a different Lindbladian.

We conjecture that all thermal of matter should be describable as steady state phases in the sense above.
However, proving this seems like a difficult task (proving mixing bounds on Lindbladians is generally a non-trivial problem), particularly for quantum Hamiltonians.

Here we show that observables in a steady state phase of matter are learnable using concentration of measure arguments.
Before starting this proof, we reiterate that each $\L_{j}(x_j)$ is supported on a ball $A_j$ around site $j\in\Lambda $ of radius $r_0$.
Each $O_i$ is supported on an area of radius no more than $k_0$.

During a training stage, we pick $N$ points $Y_1,\dots , Y_N\sim U$ independently distributed uniformly at random in the phase, which we will assume is defined by points $x\in [-1,1]^m$, and are given access to the steady state states $\rhofixed(Y_j)$. 
Next, fix $r\in\mathbb{N}$. 
Given an observable $O=\sum_{i=1}^M O_i$, we define $S_i=\operatorname{supp}(O_i)$ and 
for each $S_i$ there is a ball of diameter at most $k_0$  containing $S_i$, and further define $S_i(r):=\{j\in\Lambda|\operatorname{dist}(j,S_i)\le r\}$. 
We can choose $r_0 = O(1)$ to be the minimum integer such that the support of all $\L_j$ terms would fit inside a ball of radius $r_0$.
We construct for any $x\in [-1,1]^m$ the estimator 
\begin{align}\label{equ:defin_estimator}
		\hat{f}_O(x)=\sum_{i=1}^M\tr \big[O_i\, \rhofixed( \hat{Y}_i(x))\big]\,,\quad \text{ with }\quad \hat{Y}_i(x)=\operatorname{argmin}_{{Y}_k}\|x|_{\mathcal{S}_i(r)}-{Y}_k|_{\mathcal{S}_i(r)}\|_{\ell^\infty}\,,
	\end{align}
where we recall that we denote by $x_{\mathcal{S}_{i}(r)}$ the concatenation of vectors $x_j$ corresponding to interactions $h_j$ supported on regions intersecting $S_i(r)$.
In words, we approximate the expectation value of $O_i$ by that of the state whose parameters in a region around $S_i$ are the closest to the state of interest. 
We also denote $\mathcal{S}_i\equiv \mathcal{S}_i(0)$.
Using this we can prove the following:
\begin{proposition} \label{Prop:Approximating_Observables_Fixed}
We use the notation of the previous paragraph.
Let $\{\L(x)\}_x$ be a family of Lindbladians, each of which has a unique steady state within the same phase, defining a set $\{\rho_\infty(x)\}_x$.
Consider a set of local observables such that each acts on a ball of lattice points of diameter no more than $k_0$.
Then the estimator for $O=\sum_i O_i$, given by $\hat{f}_O(x) = \sum_i \tr[O_i \rho_\infty(\hat{Y}_i(x))]$ satisfies the bound:
\begin{align*}
    \sup_{x\in [-1,1]^m} |f_O(x) - \hat{f}_O(x)|\leq \epsilon \sum_{i_1}^M \norm{O_i}_\infty,
\end{align*}
with probability at least $(1-\delta)$, whenever
\begin{align*}
    N &= \left( \frac{\gamma}{2}\right)^{-[2(r+r_0+k_0)]^D\ell} \log\left(\frac{M}{\delta} \right) + \left( \frac{\gamma}{2} \right)^{-[2(r+r_0+k_0)]^D\ell}\log\left( \frac{\gamma}{2}\right)^{[2(r+r_0+k_0)]^D\ell}
\end{align*}
with
\begin{align*}
    r &= \left\lceil 2\xi \log\left( \frac{4c'|A| J (D-1)! (2\xi)^{D-1} D^{D-1}}{\epsilon e^{1/2\xi }(1-e^{-1/2\xi})}  \right) \right\rceil \\
    \gamma &= \frac{\epsilon}{2[2(r+k_0)]^DJ\ell}.
\end{align*}
%\textcolor{red}{These parameters are not quite correct. But correct scaling.}
\end{proposition}

Before we prove this result it is useful to study its asymptotic scaling. 
We note, $r=\Theta(\log(\epsilon^{-1}))$, $\gamma=\Theta\Big(\, \frac{\epsilon}{\log(\epsilon^{-1})^D}\Big)$, so that the number of samples needed is asymptotically
\begin{align*}
    N=\Theta \left(\log\Big(\frac{M}{\delta}\Big)\, e^{\operatorname{polylog}(\epsilon^{-1})} \right)\,.
\end{align*}

\begin{proof}[Proof of \cref{Prop:Approximating_Observables_Fixed}.]

We fix $O_i$ and $r>0$, and restrict ourselves to the subset of parameters $x|_{\mathcal{S}_i(r)}$. The number of parameters in that subset is bounded by the volume $V(r+r_0+k_0)$ of the ball $S_i(r+r_0)$ times the number $\ell$ of parameters per interaction. 
We denote this total number of parameters by $m_r:=V(r+r_0+k_0)\ell$. \
Next, we partition the parameter space $[-1,1]^{m_r}$ into cubes of side-size $\gamma\in (0,1)$. 
By the coupon collector's problem, we have that the probability that none of the sub-vectors $Y_j|_{\mathcal{S}_i(r)}$ is within one of those cubes is upper bounded by $e^{-N(\gamma/2)^{m_r}+m_r\log(2/\gamma)}$.
By a union bound, the probability that for any $i\in[M]$, any cube is visited by at least one sub-vector $Y_j|_{\mathcal{S}_i(r)}$ is lower bounded by $1-\delta$, $\delta:=Me^{-N(\gamma/2)^{m_r}+m_r\log(2/\gamma)}$. 
In other words, with probability $1-\delta$ there is a $\hat{Y}_i(x)|_{\mathcal{S}_i(r)}$ in the $N$ samples satisfying 
\begin{align}\label{xminusYSir}
\|x|_{\mathcal{S}_i(r)}-\hat{Y}_i(x)|_{\mathcal{S}_i(r)}\|_{\infty}\le \gamma
\end{align}
 for all $i\in[M]$.

Denoting $\hat{f}_{O_i}(x) = \tr[\rho(\hat{Y_i}(x)|_{S(r)})O_A]$, we now want to control the following quantity:
\begin{align}
    |f_{O_i}(x|_{S(r)}) -\hat{f}_{O_i}(x) | &\leq \lim_{t\rightarrow \infty}\tr\left[ O_i \int_0^t e^{-(t-s)\L(x|_{S(r)})}(\L(x|_{S(r)})- \L(\hat{Y_i}(x)|_{S(r)}))e^{-s\L(\hat{Y_i}(x)|_{S(r)})} (\refrho) ds\right] \nonumber \\
    &\leq \norm{\refrho}_1 \lim_{t\rightarrow \infty}  \int_0^t \norm{\L^*(x|_{S(r)})- \L^*(\hat{Y_i}(x)|_{S(r)})(O_i)} ds \\
    &\leq  \norm{O_i} \sum_{j\in S(r)} J\norm{x|_{S(r)} - \hat{Y_i}(x)|_{S(r)}}_\infty \\
    &\leq \norm{O_i}V(r)J \gamma, \label{Eq:f_local_f_hat_Bound_4}
\end{align}
where we have used \cref{Lemma:Time_Evolution_Difference} in the first line, and used \cref{xminusYSir} to go the third line to the fourth line.
From \cref{Corollary:Localised_Lindbladian_All_Times}, we have.
\begin{align*}
    |\tr[O_i(\rhofixed(x)- \rhofixed(x|_{S(r)}))]|, \ \ |\tr[O_i(\rhofixed(\hat{Y_i}(x))- \rhofixed(\hat{Y_i}(x)|_{S(r)}))]| \leq C_1e^{-r/2\xi}\norm{O_A}.
\end{align*}
Using this, the definitions $f_{O_i}(x)$, $ \hat{f}_{O_i}(x)$, and \cref{Eq:f_local_f_hat_Bound_4}:
\begin{align}
    |f_{O_i}(x) - \hat{f}_{O_i}(x)|&\leq |\tr[O_i(\rhofixed(x)- \rhofixed(x|_{S(r)}))]| + |\tr[O_i(\rhofixed(\hat{Y_i}(x))- \rhofixed(\hat{Y_i}(x)|_{S_i(r)}))]| \nonumber  \\ 
    &+ |\tr[O_i(\rhofixed(x|_{S(r)})- \rhofixed(\hat{Y_i}(x)|_{S_i(r)}))]| \nonumber \\
    &\leq (2C_1e^{-r/2\xi }+ C_2(r)\gamma)\norm{O_A}, \label{Eq:Error_Actual_Learned}
\end{align}
where $C_2(r) = V(r)J\ell$.
Now, the volume $V(s)$ of a ball of radius $s$ in $\Lambda$ is equal to 
\begin{align*}
    V(s)=\sum_{a\le s} \binom{a+D-1}{D-1}\le (2s)^D\,.
\end{align*}
Suppose we want $ |f_O(x) - \hat{f}_O(x)|\leq \epsilon$ for some $\epsilon>0$, we choose to portion the error so that $2C_1e^{-r/2\xi }\leq \epsilon/2$ and $C_2(r)\gamma\leq \epsilon/2$.
Thus we also need to choose $\gamma$ such that $C_2(r)\gamma = V(r)J\ell \gamma \leq \epsilon/2$.

Thus, given a $\delta$ we wish to achieve, we have:
\begin{align*}
    \delta:=Me^{-N(\gamma/2)^{m_r}+m_r\log(2/\gamma)}.
\end{align*}
Choosing parameters:
\begin{align*}
    r &= \left\lceil 2\xi \log\left( \frac{4C_1}{\epsilon}  \right) \right\rceil \\
    &= \left\lceil 2\xi \log\left( \frac{4c'|A| J (D-1)! (2\xi)^{D-1} D^{D-1}}{\epsilon e^{1/2\xi }(1-e^{-1/2\xi})}  \right) \right\rceil \\
    \gamma &= \frac{\epsilon}{2C_2(r)} \\
    &= \frac{\epsilon}{2[2(r+k_0)]^DJ\ell}
\end{align*}
Thus we get a number of samples needed as scaling as:
\begin{align*}
    N &=   \left( \frac{2}{\gamma}\right)^{m_r} \log\left(\frac{M}{\delta} \right) + \left( \frac{2}{\gamma} \right)^{m_r}\log\left( \frac{\gamma}{2}\right)^{m_2} \\
    &= \left( \frac{\gamma}{2}\right)^{-[2(r+r_0+k_0)]^D\ell} \log\left(\frac{M}{\delta} \right) + \left( \frac{\gamma}{2} \right)^{-[2(r+r_0+k_0)]^D\ell}\log\left( \frac{\gamma}{2}\right)^{[2(r+r_0+k_0)]^D\ell}.
\end{align*}
where we have used that $m_r =V(r+r_0+k_0) \ell = [2(r+r_0+k_0)]^D\ell$.

\end{proof}

Thus, in order to get a good approximation to $f_O(x)$, we simply need sufficiently many samples of the observable of interest at different points in the parameter space of the phase.
We can learn these using classical shadows.
Since only one copy of each $\widetilde{\rhofixed}(\hat{Y}_i(x))$ is available, the value of observables reconstructed from it is likely to be too noisy. 
Here we follow the same steps as \cite{Onorati_Rouze_Stilch_Franca_Watson_2023} and use multiple copies of states ``close to'' $x$ to reconstruct a better shadow at $x$.

More specifically, consider a steady state from the phase $\rho(x)$ and a family $\rhofixed(x_1),\dots,\rhofixed(x_N)$ of states with the promise that for any $i\in [M]$ there exist $t$ vectors $x_{i_1},\dots, x_{i_t}$ such that $\max_{j\in[t]}\|x|_{\mathcal{S}_i(r)}-x_{i_j}|_{\mathcal{S}_i(r)}\|_\infty \le \gamma$. 
We run the shadow protocol and construct product operators $\widetilde{\rhofixed}( x_1),\dots,\widetilde{\rhofixed}(x_N)$. 
Then for any ball $B$ of radius $k_0$, we select the shadows $\widetilde{\rhofixed}(x_{i_1}),\dots \widetilde{\rhofixed}(x_{i_t})$ and construct the empirical average
\begin{align} \label{Eq:Aggregated_Shadow}
    \widetilde{\rho}_{B}(x):=\frac{1}{t}\sum_{j=1}^t\,\tr_{B^c}\big[\widetilde{\rhofixed}(x_{i_j})\big]\,.
\end{align}

\noindent The error of this estimate can be bound using the following from \cite{Onorati_Rouze_Stilch_Franca_Watson_2023}:

\begin{proposition}[Robust shadow tomography, Prop D.2 of \cite{Onorati_Rouze_Stilch_Franca_Watson_2023}]\label{Prop:Robust_Shadows}

Fix $\epsilon,\delta\in(0,1)$. 
In the notations of \Cref{Prop:Approximating_Observables_Fixed}, with probability $1-\delta'$, for any ball $B$ of radius $k_0$, the shadow $\widetilde{\rho}_{B}(x)$ satisfies $\|\widetilde{\rho}_{B}(x)-\tr_{B^c}[\rhofixed(x)]\|_1\le 2C_1\,e^{-\frac{r}{2\xi}}+C_2(r)\gamma+\epsilon$  as long as
\begin{align}\label{ttimes}
    q\ge \frac{8.12^{k_0}}{3.\epsilon^2}\log\left(\frac{n^{k_0}2^{k_0+1}}{\delta'}\right)\,.
\end{align}
\end{proposition}

We are now ready to prove the main result of this section --- this is the formal version of \cref{Theorem:Main_Result_Fixed_Points}:
\begin{theorem}[Learning Steady State Phases]

Use the notation of \cref{Prop:Approximating_Observables_Fixed} and \cref{Eq:Aggregated_Shadow}.
    Denote $\widetilde{f}_{O_i}(x)=\tr\big[O_i\, \widetilde{\rho}_{S_i}(x)\big]$ the function constructed from the shadow tomography protocol of \Cref{Prop:Robust_Shadows}.
    Consider a set of $N$ shadows $\{\widetilde{\rhofixed}(x_i) \}_{i=1}^N$ and suppose we wish that:
    \begin{align*}
        |f_{O}(x)-\widetilde{f}_O(x)|\le  \epsilon\,\sum_{i}\,\|O_i\|_\infty\,,
    \end{align*}
with probability $(1-\delta).(1-\delta')$, with associated parameters:
\begin{align*}
    r &=  2\xi \log\left( \frac{4c'|A| J (D-1)! (2\xi)^{D-1} D^{D-1}}{\epsilon e^{1/2\xi }(1-e^{-1/2\xi})}  \right)  \\
    \gamma &= \frac{\epsilon}{2[2(r+k_0)]^DJ\ell} \\
    q &= \frac{8.12^{k_0}}{3.\epsilon^2}\log\left(\frac{n^{k_0}2^{k_0+1}}{\delta'}\right)\,.
\end{align*}
Then it is sufficient to choose:
\begin{align*}
    N= q\left( \frac{\gamma}{2}\right)^{-[2(r+r_0+k_0)]^D\ell} \log\left(\frac{M}{\delta} \right) + q\left( \frac{\gamma}{2} \right)^{-[2(r+r_0+k_0)]^D\ell}\log\left( q\frac{\gamma}{2}\right)^{[2(r+r_0+k_0)]^D\ell}
\end{align*}
\end{theorem}

\begin{proof}
    We follow the reasoning of \cite[Theorem D.3]{Onorati_Rouze_Stilch_Franca_Watson_2023}:
    adapting the proof of \cref{Prop:Approximating_Observables_Fixed}, it is clear that with probability  $$1-\delta:=1-Me^{-N\frac{1}{q}(\gamma/2)^{m_r}+m_r\log(2/\gamma)+\log q}$$
each cube is visited at least $t$ times. Conditioned on that event, and choosing $q$ such that \Cref{ttimes} holds, we have that with probability $1-\delta'$ 
\begin{align}\label{Eq:Full_Error_Expression}
    |f_{O_i}(x)-\widetilde{f}_{O_i}(x)|\le \Big(2C_1 e^{-\frac{r}{2\xi}}+C_2(r)\gamma +\epsilon \Big)\|O_i\|_\infty\,.
\end{align}
\end{proof}

\subsection{Learning General Phases of Matter with Local Rapid Mixing}

\Cref{Sec:Fixed-Point_Phases} proved that we could efficiently learn a phase of matter which is restricted to steady states of Lindbladians. 
However, there are points in the phase which are not steady states of Lindbladians, but can nonetheless be reached by a rapid time evolution under a Lindbladian with some auxiliary state.
We wish to show that we can learn these phases too.
We first need to prove an analogue of \cref{Corollary:Localised_Lindbladian_All_Times}, but where the times are not necessarily equal:

\begin{lemma} 
    Let $\L(x) = \L + \sum \L_j(x_j)$ be a continuously parameterised family of Lindbladians, where each $\L_j(x_j)$ is a local term supported on $O(1)$ qudits and $x_j\in [-1,1]^\ell$ for $\ell = O(1)$.
    Then 
    \begin{align*}
        |f_O(\L(x),\omega,t) - f_O(\L(x|_{S(r)}),\omega,t')| \leq \norm{O_A}\left( 2C_1 e^{-r/2\xi} + J|S(r)|\delta t\right),
    \end{align*}
    for an $O(1)$ constant $C_1$ and where $|t-t'|\leq \delta t$.
\end{lemma}
\begin{proof}
Since $\omega$ will be fixed throughout, we will neglect it for now and write $f_O(\L(x),t) \coloneqq f_O(\L(x),\omega,t)$.
    Using the triangle inequality:
    \begin{align*}
        |f_O(\L(x),t) - f_O(\L(x|_{S(r)}),t')| &\leq |f_O(\L(x),t) - f_O(\L(x|_{S(r)}),t)| + |f_O(\L(x|_{S(r)}),t) - f_O(\L(x|_{S(r)}),t')|.
    \end{align*}
    From \cref{Corollary:Localised_Lindbladian_All_Times}, the first term on the right-hand side is bounded as:
    \begin{align*}
        |f_O(\L(x),t) - f_O(\L(x|_{S(r)}),t)| \leq C_1 e^{-r/2\xi}\norm{O_A}.
    \end{align*}
    For the second term on the right-hand side, assume $t'>t$ without loss of generality, and consider:
    \begin{align*}
        f_O(\L(x|_{S(r)}),t') - f_O(\L(x|_{S(r)}),t) &= \tr \int_t^{t'} O_A\L(x|_{S(r)})e^{s\L(x|_{S(r)})}(\rho_0) ds \\
        &\leq\int_t^{t'} \norm{O_A}\norm{\L(x|_{S(r)})}_{1\rightarrow 1, cb} \norm{e^{s\L(x|_{S(r)})}}_{1\rightarrow 1, cb}  \norm{\rho_0}_1 ds \\  
        &\leq (t-t')\norm{O_A} J |S(r)|.
    \end{align*}
    Taking into account the case $t'\leq t$, we get 
    \begin{align*}
        |f_O(\L(x|_{S(r)}),t') - f_O(\L(x|_{S(r)}),t)| \leq J\delta t \norm{O_A}|S(r)|.
    \end{align*}
    Combining these two gives:
    \begin{align*}
        |f_O(\L(x),t) - f_O(\L(x|_{S(r)}),t')| \leq \norm{O_A}\left( 2C_1 e^{-r/2\xi} + J|S(r)|\delta t\right)
    \end{align*}
\end{proof}

\noindent We now move onto learning the phase of matter.
Similar to learning the steady state phase as outlined in \cref{Sec:Fixed-Point_Phases}, we will take $N$ samples, but noting that time must be considered as an additional parameter, they are labelled as $(Y_1, \tau_1), \dots (Y_N, \tau_N) \sim U$.
We are given access to the corresponding states $e^{\tau_i(t)\L(\hat{Y}_i(x))} (\rho^*\otimes \omega)$. 
We can then construct estimators:
\begin{align}
    \hat{f}(\L(x), \omega, t) &= \sum_j \tr[O_j 
 e^{\tau_i(t)\L(\hat{Y}_i(x))} (\rho^*\otimes \omega)] \\ \text{ with }\quad (\hat{Y}_i(x), \tau_i(t)) &\coloneqq \operatorname{argmin}_{({Y}_k, \tau_k)}\| (x|_{\mathcal{S}_i(r)}, t)-({Y}_k|_{\mathcal{S}_i(r)}, \tau_k)\|_{\ell^\infty}\,, 
\end{align}
where in the second expression, we treat $(x|_{\mathcal{S}_i(r)}, t)$ as a concatenated vector of the vector $x|_{\mathcal{S}_i(r)}$ and $t$.
We remember that $\refrho$ is the reference state for the phase.

In words, we approximate the expectation value of $O_i$ by that of the state whose parameters in a region around $S_i$ are the closest to the state of interest \textbf{and} who are close in time. 
We can then bound the error associated with this estimator as follows:
\begin{proposition}[Rapid Local Mixing: Learning General Lindbladian Phases] \label{Prop:Approximating_Observables_General}
Consider a phase of matter.
Then the estimator $\hat{f}(\L(x), \omega, t) = \sum_j \tr[O_j 
 e^{\tau_i(t)\L(\hat{Y}_i(x))} (\rho^*\otimes \omega)]$ satisfies the bound:
\begin{align*}
    \sup_{\substack{x\in [-1,1]^m\\ t\geq 0 \\ \omega \in W}} |f_O(\L(x), \omega,  t)) - \hat{f}_O(\L(x), \omega,  t))|\leq \epsilon \sum_{i_1}^M \norm{O_i}_\infty,
\end{align*}
with probability at least $(1-\delta)$, whenever
\begin{align*}
    N &= |W|\left( \frac{\gamma}{t_\epsilon}\right)^{-1}\left( \frac{\gamma}{2}\right)^{-[2(r+r_0+k_0)^D\ell]} \left[  \log\left(\frac{M}{\delta} \right) + [2(r+r_0+k_0)^D\ell]\log\left( \frac{\gamma}{2}\right) + \log\left( \frac{\gamma}{t_\epsilon}\right) \right].
\end{align*}
with
\begin{align*}
    r &=  2\xi \log\left( \frac{4c'|A| J (D-1)! (2\xi)^{D-1} D^{D-1}}{\epsilon e^{1/2\xi }(1-e^{-1/2\xi})}  \right)  \\
    \gamma &= \frac{\epsilon}{2[2(r+k_0)]^DJ\ell} \\
    t_\epsilon &= \frac{1}{\gamma'}\log\left(\frac{6c'(|A|)}{\epsilon}  \right).
\end{align*}

\end{proposition}

We note that if we take the asymptotic scaling, then $r=O(\log(1/\epsilon))$, $\gamma = O(\epsilon/\log^D(1/\epsilon))$ and $t_\epsilon = O(\log(1/\epsilon))$, which together gives a leading order term of:
\begin{align*}
    N = O\left( \log\left(\frac{M}{\delta} \right) 2^{-\textrm{polylog}(\epsilon^{-1})}  \right ).
\end{align*}

\begin{proof}[Proof of \cref{Prop:Approximating_Observables_General}.]
    We wish to be able to learn $f_O$ everywhere in the phase.
    However, while the parameter $x\in [-1,1]^m$ has a bounded domain, $t$ does not.
    We can fix this by realising we can get $\epsilon$ close to any state using \cref{Assumption:Local_Rapid_Mixing}:
         \begin{align*}
       | \tr[O_A(T_t(\rho_0) - \rho_\infty)]| \leq c'(|A|) e^{-\gamma' t}
    \end{align*}
    hence if we wish to get precision $\epsilon_{mix}$, it is sufficient to choose
    \begin{align} \label{Eq:eps_mix_definition}
        t_\epsilon = \frac{1}{\gamma'}\log\left(\frac{2c'(|A|)}{\epsilon_{mix}}  \right).
    \end{align}
    Thus to learn $f_O$ everywhere to high precision, we need only consider $t\in [0,t_\epsilon]$.

We fix $O_i$ and $r>0$, and restrict ourselves to the subset of parameters $x|_{\mathcal{S}_i(r)}$ and the time parameter $t$.
The number of parameters in that subset is bounded by the volume $V(r+r_0+k_0)$ of the ball $S_i(r+r_0)$ times the number $\ell$ of parameters per interaction. 
We denote it by $m_r:=V(r+r_0+k_0)\ell$. 
Next, we partition the parameter space $[-1,1]^{m_r}\times [0,t_\epsilon]$ onto cubes of side-length $\gamma\in (0,1)$. 
By the coupon collector's problem, we have that the probability that none of the sub-vectors $Y_j|_{\mathcal{S}_i(r)}$ is within one of those cubes is upper bounded by $e^{-N(\gamma/2)^{m_r}(\gamma/t_\epsilon)+m_r\log(2/\gamma)+\log(t_\epsilon/\gamma)}$.
By union bound, the probability that for any $i\in[M]$, any cube is visited by at least one sub-vector $Y_j|_{\mathcal{S}_i(r)}$ is lower bounded by $1-\delta$, $\delta:=Me^{-N(\gamma/2)^{m_r}(\gamma/t_\epsilon)+m_r\log(2/\gamma)+\log(t_\epsilon/\gamma)}$. 
In other words, with probability $1-\delta$ there is a $\hat{Y}_i(x)|_{\mathcal{S}_i(r)}$ in the $N$ samples satisfying 

\begin{align}\label{Eq:xminusYSir}
    \max\{ \|x|_{\mathcal{S}_i(r)}-\hat{Y}_i(x)|_{\mathcal{S}_i(r)}\|_{\infty} \ , \ |t-\tau_i| \}\le \gamma
\end{align}
 for all $i\in[M]$. 

 We wish to bound the following quantity:
 \begin{align} 
     |f_O(\L(x), \omega,  t)) - \hat{f}_O(\L(x), \omega,  t))| 
     &\leq |f_O(\L(x),t) -f_O(\L(x|_{S(r)}),t) | \label{Eq:Space_Time_Parameters} \\
     &+|f_O(\L(x|_{S(r)}),t) -f_O(\L(\hat{Y}_i(x)|_{S(r)}),t) | \label{Eq:Term_2}\\
    &+ |f_O(\L(\hat{Y}_i(x)|_{S(r)}),t)  -\hat{f}_O(\L(x),t) |, \label{Eq:Term_3} 
 \end{align}
 where we have used the triangle inequality.

%We now wish to control the following quantity:
%\begin{align} 
%    |f_O(\L(x|_{S(r)}),t) -\hat{f}_O(\L(x),t) | \leq& |f_O(\L(x|_{S(r)}),t) -f_O(\L(\hat{Y}_i(x)|_{S(r)}),t) | \\
%    &+ |f_O(\L(\hat{Y}_i(x)|_{S(r)}),t)  -\hat{f}_O(\L(x),t) |   
%\end{align}
Considering the second term, \cref{Eq:Term_2}, on the right-hand side and applying  \cref{Lemma:Time_Evolution_Difference}:
    \begin{align*}
  |f_O(\L(x|_{S(r)}),t) -f_O(\L(\hat{Y}_i(x)|_{S(r)}),t) |  \leq& \bigg| \tr\bigg[ O_A \int_0^t e^{-(t-s)\L(x|_{S(r)})}(\L(x|_{S(r)}) \\ &- \L(\hat{Y_i}(x)|_{S(r)}))e^{-s\L(\hat{Y_i}(x)|_{S(r)})} (\refrho) ds\bigg] \bigg| \\
    %&\leq \norm{\refrho}_1 \ \int_0^t \norm{\L^*(x|_{S(r)})- \L^*(\hat{Y_i}(x)|_{S(r)})(O_A)} ds \\
    \leq& \bigg| \sum_{u} \tr\left[ \refrho \int_0^{t} S_{t-s}^* E_{u,r}^*T_s^*(O_A) ds     \right] \bigg| \\
     \leq& \sum_{u}|x_{u,r}-\hat{y}_{u,r}| \bigg| \tr\left[ \refrho \int_0^{t} S_{t-s}^* \L_{u,r}^* T_s^*(O_A) ds   \bigg|  \right],
    \end{align*}
    where here we have used $x_{u,r},\hat{y}_{u,r}$ to denote the elements of $x, \hat{Y}_i$ which describe the term $\L_{u,r}$, and we have that the Lindbladian is linearly parameterised.
    We then bound this as per \cref{Eq:Difference_Lindbladian_Single_Term}:
    \begin{align*}
    &\leq  c(|A|)\norm{O_A} \sum_{i\in S(r)} J\norm{x|_{S(r)} - \hat{Y_i}(x)|_{S(r)}}_{\ell_\infty} \\
    &\leq c(|A|)\norm{O_A} V(r)J \gamma.
\end{align*}
For the third term on the right-hand side \cref{Eq:Term_3},
    $|f_O(\L(\hat{Y}_i(x)|_{S(r)}),t)  -\hat{f}_O(\L(x),t) | $
assume $\tau_i(t)>t$ without loss of generality, and consider: 
    \begin{align*}
    |f_O(\L(\hat{Y}_i(x)|_{S(r)}),t)  -\hat{f}_O(\L(x),t)| &= |f_O(\L(\hat{Y}_i(x)|_{S(r)}),t)  -f_O(\L(\hat{Y}_i(x)|_{S(r)}),\tau_i(t))|  \\
         &= \tr \int_t^{\tau_i(t)} O_A\L(x|_{S(r)})e^{s\L(x|_{S(r)})}(\refrho) ds \\
        &\leq\int_t^{\tau_i(t)} \norm{O_A}\norm{\L(x|_{S(r)})}_{1\rightarrow 1, cb} \norm{e^{s\L(x|_{S(r)})}}_{1\rightarrow 1, cb}  \norm{\refrho}_1 ds \\  
        &\leq (\tau_i(t)-t)\norm{O_A} J |S(r)|.
    \end{align*}
    Taking into account the case $\tau_i(t)\leq t$, we get 
    \begin{align*}
         |f_O(\L(\hat{Y}_i(x)|_{S(r)}),t)  -\hat{f}_O(\L(x),t)| &\leq J |\tau_i(t)- t| \norm{O_A}|S(r)| \\
         &\leq J \gamma \norm{O_A}|S(r)|. 
    \end{align*}
    where the last line follows from $|\tau_i-t|\leq \gamma$ as per \cref{Eq:xminusYSir}.
    Thus using \cref{Eq:Space_Time_Parameters} we can bound the total difference as:
    \begin{align*}
        |f_O(\L(x|_{S(r)}),t) -\hat{f}_O(\L(x),t) | &\leq \norm{O_A} V(r)J\ell \gamma + \gamma J  \norm{O_A}|S(r)|
    \end{align*}
    Finally, if we make the restriction that $\tau_i(t)\leq t_\epsilon$, then we have additional error from this assumption (see \cref{Eq:eps_mix_definition} for definition of $\epsilon_{mix}$):
     \begin{align*}
        |f_O(\L(x|_{S(r)}),t) -\hat{f}_O(\L(x),t) | &\leq \norm{O_A}V(r)J\ell \gamma + \gamma J  \norm{O_A}|S(r)| + \epsilon_{mix} \\
        &= \gamma\norm{O_A}V(r)J\ell  + \gamma J  \norm{O_A}V(r) + \epsilon_{mix} \\
        &= \gamma\norm{O_A}V(r)J(\ell+1) + c(|A|)\norm{O_A}e^{-\gamma' t_\epsilon} \\
        &\leq \gamma\norm{O_A}J(\ell+1)(2r)^D + c(|A|)\norm{O_A}e^{-\gamma' t_\epsilon}.
    \end{align*}
    We now consider the total error according to \cref{Eq:Space_Time_Parameters}. 
    First consider: 
    \begin{align*}
    |f_O(\L(x),t) - f_O(\L(x|_{S(r)}),t)|, \ \ |\tr[O_A(\rho(\hat{Y_i}(x),t)- \rho(\hat{Y_i}(x)|_{S(r)}),t)]| \leq C_1e^{-r/2\xi}\norm{O_A},
\end{align*}
where the bound is given by \cref{Corollary:Localised_Lindbladian_All_Times}.
Using this and applying the bounds to \cref{Eq:Space_Time_Parameters}, we have: 
\begin{align*}
    |f_O(\L(x),t) -\hat{f}_O(\L(x),t) | &\leq   2C_1 e^{-r/2\xi} \norm{O_A} +  \gamma\norm{O_A}J(\ell+1)(2r)^D + c(|A|)\norm{O_A}e^{-\gamma' t_\epsilon}.
\end{align*}
We arbitrarily choose the error budget such that:
\begin{align*}
    2C_1 e^{-r/2\xi} \norm{O_A}, \; \gamma\norm{O_A}J(\ell+1)(2r)^D, \; c(|A|)\norm{O_A}e^{-\gamma' t_\epsilon}\leq \epsilon/3.
\end{align*}
This allows us to choose:
\begin{align*}
    r  &= 2\xi \log\left( \frac{24C_1}{\epsilon}  \right) \\
    \gamma &=  \frac{\epsilon}{J(\ell+1)(2r)^D}  \\
    t_\epsilon &= \frac{1}{\gamma'}\log \left( \frac{3c(|A|)}{\epsilon} \right).
\end{align*}

\noindent Doing so gives us a number of samples scaling as:
\begin{align*}
    N &=  \left( \frac{t_\epsilon}{\gamma}\right)\left( \frac{2}{\gamma}\right)^{m_r} \left[  \log\left(\frac{M}{\delta} \right) + \log\left( \frac{\gamma}{2}\right)^{m_r} + \log\left( \frac{\gamma}{t_\epsilon}\right) \right] \\
    &= \left( \frac{\gamma}{t_\epsilon}\right)^{-1}\left( \frac{\gamma}{2}\right)^{-[2(r+r_0+k_0)\ell]} \left[  \log\left(\frac{M}{\delta} \right) + [2(r+r_0+k_0)\ell]\log\left( \frac{\gamma}{2}\right) + \log\left( \frac{\gamma}{t_\epsilon}\right) \right].
\end{align*}

\noindent Finally, for need to repeat this process for every $\omega_i\in W$.
However, since $|W|=O(1)$, this only needs to be done a constant number of times.
This gives the lemma statement.
 
\end{proof}

We now have results showing that we learn phases of matter if we have estimates of observables to sufficient precision everywhere. 
We now need to show that we can generate such estimates efficiently. 
To do so, we again enlist the robust classical shadows technique discussed previously, and formalised in \cref{Prop:Robust_Shadows}.

Consider a state from the phase $\rho(x, \omega,  \tau)$ and a family $\rho(x_1, \omega, \tau_1),\dots,\rho(x_N\omega, \tau_N)$ of states with the promise that for any $i\in [M]$ there exist $t$ vectors $(x_{i_1}, \tau_{i_1}),\dots, (x_{i_t}, \tau_{i_t})$ such that $\max_{j\in[t]}\| (x|_{\mathcal{S}_i(r)}, \tau)-(x_{i_j}|_{\mathcal{S}_i(r)}, \tau_{i_j})\|_\infty \le \gamma$. 
We run the shadow protocol and construct product operators $\widetilde{\rho}( x_1, \omega, \tau_1),\dots,\widetilde{\rho}(x_N, \omega, \tau_N)$. 
Then for any ball $B$ of radius $k_0$, we select the shadows $\widetilde{\rho}(x_{i_1}, \omega, \tau_{i_1}),\dots \widetilde{\rho}(x_{i_t}, \omega, \tau_{i_t})$ and construct the empirical average
\begin{align} \label{Eq:Aggregated_Shadow_2}
    \widetilde{\rho}_{B}(x,\omega, \tau):=\frac{1}{q}\sum_{j=1}^q\,\tr_{B^c}\big[\widetilde{\rho}(x_{i_j}, \omega, \tau_{i_j})\big]\,.
\end{align}

\noindent Using this we can state our main result of this work (the formal version of \cref{Theorem:Main_Result}).

\begin{theorem}[Learning Phases of Matter] \label{Theorem:Learning_General_Phase}
Use the notation of \cref{Prop:Approximating_Observables_General}.
    Denote $\tilde{f}_{O_i}(\L(x), \omega,  t))=\tr\big[O_i\, \widetilde{\rho}_{S_i}(x,\omega, t)\big]$ the function constructed from the shadow tomography protocol of \Cref{Prop:Robust_Shadows}.
    Consider a set of $N$ shadows $\{\widetilde{\rho}(x_i, ,\omega_i, \tau_i) \}_{i=1}^N$ and suppose we wish that:
    \begin{align*}
    \sup_{\substack{x\in [-1,1]^m \\ t\geq 0 \\ \omega \in W}} |f_O(\L(x), \omega,  t)) - \tilde{f}_O(\L(x), \omega,  t))|\leq \epsilon \sum_{i_1}^M \norm{O_i}_\infty,
\end{align*}
%    \begin{align*}
%        |f_{O}(x)-\widetilde{f}_O(x)|\le  %\epsilon\,\sum_{i}\,\|O_i\|_\infty\,,
%    \end{align*}
with probability $(1-\delta).(1-\delta')$, with associated parameters:
\begin{align*}
     r &=  2\xi \log\left( \frac{4c'|A| J (D-1)! (2\xi)^{D-1} D^{D-1}}{\epsilon e^{1/2\xi }(1-e^{-1/2\xi})}  \right)  \\
    \gamma &= \frac{\epsilon}{2[2(r+k_0)]^DJ\ell} \\
    t_\epsilon &= \frac{1}{\gamma'}\log\left(\frac{6c'(|A|)}{\epsilon}  \right) \\
    q &= \frac{8.12^{k_0}}{3.\epsilon^2}\log\left(\frac{n^{k_0}2^{k_0+1}}{\delta'}\right)\,.
\end{align*}
Then it is sufficient to choose:
\begin{align*}
    N=  |W|q\left( \frac{\gamma}{t_\epsilon}\right)^{-1}\left( \frac{\gamma}{2}\right)^{-[2(r+r_0+k_0)\ell]} \left[  \log\left(\frac{M}{\delta} \right) + [2(r+r_0+k_0)\ell]\log\left( \frac{\gamma}{2}\right) + \log(q) + \log\left( \frac{\gamma}{t_\epsilon}\right) \right].
\end{align*}
\end{theorem}

\begin{proof}
    We follow the reasoning of \cite[Theorem D.3]{Onorati_Rouze_Stilch_Franca_Watson_2023}:
    adapting the proof of \cref{Prop:Approximating_Observables_General}, it is clear that with probability  $$1-\delta:=1-Me^{-N\frac{1}{q}(\gamma/2)^{m_r}(\gamma/t_{\epsilon})+m_r\log(2/\gamma)+\log(t_\epsilon/\gamma)+\log q}$$
each cube is visited at least $q$ times. 
Conditioned on that event, and choosing $t$ such that \Cref{ttimes} holds, we have that with probability $1-\delta'$ 
\begin{align*}
    |f_O(\L(x),t) -\tilde{f}_O(\L(x),t) | &\leq   C_1 e^{-r/2\xi} \norm{O_A} +  \gamma\norm{O_A}J(\ell+1)(2r)^D + c(|A|)\norm{O_A}e^{-\gamma' t_\epsilon}.
\end{align*}
We then choose the parameters to achieve the error in the theorem statement.
\end{proof}

\section{Learning Families of States with Slow Mixing}

%\textcolor{red}{Again, we have neglected auxiliary states here.}

We can also consider the case where rapid local mixing does not occur \textbf{and} we allow for slower mixing that the $\polylog(n)$ global mixing. 
That is, we will replace the mixing assumption with 
\begin{align}
    \norm{ e^{t\mathcal{L}}(\rho_0\otimes \omega_0) - \rho_1\otimes \omega_1    }_1 \leq f(n)e^{-\gamma't}.
\end{align}
As a result, for a general local observable we have that the local mixing is only as fast as the global mixing between states in the same phase:
    \begin{align}\label{Eq:Slow_Local_Mixing}
       | \tr[O_A(T_t(\rho_0\otimes \omega_0) - \rho_\infty\otimes \omega _1)]| \leq \norm{O_A} f(n)e^{-\gamma' t}.
    \end{align}
Similarly to the rapidly mixing case, we can consider the set of steady states under this slow mixing condition.
We call such a set of states a \textit{Steady State Phase with Slow Mixing} (as they satisfy a similar condition the Lindbladian definition of phase).

\subsection{Slow Mixing: Steady State Phases}

We first consider learning steady state phases, as defined in \cref{Def:Fixed_Point_Phase}.
\noindent Implementing the slow mixing condition in the proof of \cref{Lemma:Time_Evolution_Difference} gives a weaker bound:
\begin{lemma}[Slow Local Mixing: Perturbation Bounds] \label{Lemma:Slow_Local_Time_Evolution}
    Let $\L$ be a Lindbladian, and $\L'=\L+E$ another local Lindbladian, where
    \begin{align*}
        E = \sum_{u,r}E_{u,r}
    \end{align*}
    and where $\L,\L'$ the steady states of both Lindbladians are in the same phase.
    Then:
    \begin{align*}
        |f_O(\L,t) - f_O(\L',t)| \leq c(n)\norm{O_A}\norm{E^*}_{\infty\rightarrow \infty, cb}.
    \end{align*}
    where $c(n)=O(f(n))$.
\end{lemma}
\begin{proof}[Proof Sketch]
    We follow the proof of \cref{Lemma:Time_Evolution_Difference} until we reach \cref{Eq:Second_Int_Line_2}.
    At this point we use the bound \cref{Eq:Slow_Local_Mixing}.
    Continuing from here gives the lemma statement.
\end{proof}

\noindent With this lemma in hand, we can follow through the same proof as for the rapid mixing proof.
In particular, in analogy to \cref{Corollary:Localised_Lindbladian_All_Times}, we get the following: 
\begin{corollary}[Slow Local Mixing: Localising Functions] \label{Corollary:Slow_Mixing_Localised_Lindbladian_All_Times}
    Let $\L(x) = \L + \sum \L_j(x_j)$ be a continuously parameterised family of Lindbladians, where each $L_j(x_j)$ is a local term supported on $O(1)$ qudits and $x_j\in [-1,1]^\ell$ for a constant $\ell$.
    Then 
    \begin{align*}
        |f_O(\L(x),t) - f_O(\L(x|_{S(r)}),t)| \leq C_1(n) e^{-r/2\xi}\norm{O_A},
    \end{align*}
    for an $C_1(n) = O(f(n))$.
\end{corollary}

Pushing these changes through the lemmas about learning --- in particular \cref{Prop:Approximating_Observables_Fixed}, the following can be shown:

\begin{theorem}[Slow Local Mixing: Learning steady state Phases]
Use the notation of \cref{Prop:Approximating_Observables_Fixed} and \cref{Eq:Aggregated_Shadow_2}.
    Denote $\widetilde{f}_{O_i}(x)=\tr\big[O_i\, \widetilde{\rho}_{S_i}(x)\big]$ the function constructed from the shadow tomography protocol of \Cref{Prop:Robust_Shadows}.
    Consider a set of $N$ shadows $\{\widetilde{\rho}(\beta,x_i) \}_{i=1}^N$ and suppose we wish that:
    \begin{align*}
        |f_{O}(x)-\widetilde{f}_O(x)|\le  \epsilon\,\sum_{i}\,\|O_i\|_\infty\,,
    \end{align*}
with probability $(1-\delta).(1-\delta')$, with associated parameters:
\begin{align*}
    r &=  2\xi \log\left( \frac{4c'|A| J (D-1)! (2\xi)^{D-1} D^{D-1} f(n)}{\epsilon e^{1/2\xi }(1-e^{-1/2\xi})}  \right)  \\
    \gamma &= \frac{\epsilon}{2[2(r+k_0)]^DJ\ell} \\
    q &= \frac{8.12^{k_0}}{3.\epsilon^2}\log\left(\frac{n^{k_0}2^{k_0+1}}{\delta'}\right)\,.
\end{align*}
Then it is sufficient to choose:
\begin{align*}
    N&= q\left( \frac{\gamma}{2}\right)^{-[2(r+r_0+k_0)]^D\ell} \log\left(\frac{M}{\delta} \right) + q\left( \frac{\gamma}{2} \right)^{-[2(r+r_0+k_0)]^D\ell}\log\left( q\frac{\gamma}{2}\right)^{[2(r+r_0+k_0)]^D\ell} \\
    N&= O\bigg(  2^{\polylog(f(n)/\epsilon)}\log\bigg(\frac{M}{\delta}\bigg)\log\bigg(\frac{n}{\delta}\bigg)\bigg).
\end{align*}
\end{theorem}
\begin{proof}
     It is clear that with probability  $$1-\delta:=1-Me^{-N\frac{1}{q}(\gamma/2)^{m_r}+m_r\log(2/\gamma)+\log t}$$
    each cube is visited at least $q$ times.
    Thus the appropriate value of $N$ is given by:
    \begin{align*}
         N= q\left( \frac{\gamma}{2}\right)^{-m_r} \log\left(\frac{M}{\delta} \right) + q\left( \frac{\gamma}{2} \right)^{-m_r}\log\left( q\frac{\gamma}{2}\right)^{m_r}. 
    \end{align*}
    where $m_r = [2(r+r_0+k_0)]^D\ell$.

    To make an appropriate choice of $r,\gamma$, we need to understand the error between the localised function and the full version.
    We we use an analogy to \cref{Eq:Error_Actual_Learned}:
    \begin{align*}
    |f_{O_i}(x)-\widetilde{f}_{O_i}(x)|\le \Big(2C_1(n) e^{-\frac{r}{2\xi}}+C_2(r)\gamma +\epsilon \Big)\|O_i\|_\infty\,
    \end{align*}
    where $C_1(n)=O(f(n))$ as per \cref{Corollary:Slow_Mixing_Localised_Lindbladian_All_Times} and $C_2(r)=2[2(r+k_0)]^DJ\ell$ is defined as per the proof of \cref{Prop:Approximating_Observables_Fixed}. 
    Setting these to be of the relevant error size, we see that we can choose:
    \begin{align*}
        r &= 2 \xi \log\left( \frac{C_1(n)}{\epsilon} \right) =   2\xi \log\left( \frac{4c'|A| J (D-1)! (2\xi)^{D-1} D^{D-1}f(n)}{\epsilon e^{1/2\xi }(1-e^{-1/2\xi})}  \right) \\
        \gamma &= \frac{\epsilon}{C_2(r)}  = \frac{\epsilon }{2[2(r+k_0)]^DJ\ell} = O\bigg( \frac{\epsilon}{\log^D(f(n)/\epsilon)} \bigg) = O\bigg( \epsilon2^{-c\log\log(f(n)/\epsilon)} \bigg)
    \end{align*}
    
\end{proof}

\subsection{Slow Local Mixing: General Lindbladian Phases}

In the case of a general phase of matter as per \cref{Def:Phase_of_Matter}, we see that we can repeat the same analysis.
Using \cref{Corollary:Slow_Mixing_Localised_Lindbladian_All_Times}, we can derive a bound on the number of samples needed in a similar way to \Cref{Theorem:Learning_General_Phase}.
In particular, we only need to be careful about the time parameter when learning the entire phase (rather than just the steady states).
We use the same learning algorithm as in the rapid local mixing case.
\begin{theorem}[Slow Local Mixing: Learning Phases of Matter] \label{Theorem:Slow_Mixing:Learning_General_Phase}
Use the notation of \cref{Prop:Approximating_Observables_General}.
    Denote  $\tilde{f}_{O_i}(\L(x), \omega,  t))=\tr\big[O_i\, \widetilde{\rho}_{S_i}(x,\omega, t)\big]$ the function constructed from the shadow tomography protocol of \Cref{Prop:Robust_Shadows}.
    Consider a set of $N$ shadows $\{\widetilde{\rho}(x_i, \omega_i, \tau_i) \}_{i=1}^N$ and suppose we wish that:
    \begin{align*}
    \sup_{\substack{x\in [-1,1]^m \\ t\geq 0 \\ \omega \in W}} |f_O(\L(x), \omega,  t)) - \tilde{f}_O(\L(x), \omega,  t))|\leq \epsilon \sum_{i_1}^M \norm{O_i}_\infty,
\end{align*}
%    \begin{align*}
%        |f_{O}(x)-\widetilde{f}_O(x)|\le  %\epsilon\,\sum_{i}\,\|O_i\|_\infty\,,
%    \end{align*}
with probability $(1-\delta).(1-\delta')$, with associated parameters:
\begin{align*}
     r &=  2\xi \log\left( \frac{4c'|A| J (D-1)! (2\xi)^{D-1} D^{D-1}f(n)}{\epsilon e^{1/2\xi }(1-e^{-1/2\xi})}  \right)  \\
    \gamma &= \frac{\epsilon}{3[2(r+k_0)]^DJ(\ell+1)} \\
    t_\epsilon &= \frac{1}{\gamma'}\log\left(\frac{3f(n)}{\epsilon}  \right) \\
    q &= \frac{8.12^{k_0}}{3.\epsilon^2}\log\left(\frac{n^{k_0}2^{k_0+1}}{\delta'}\right)\,.
\end{align*}
Then it is sufficient to choose:
\begin{align*}
    N&=  |W|q\left( \frac{\gamma}{t_\epsilon}\right)^{-1}\left( \frac{\gamma}{2}\right)^{-[2(r+r_0+k_0)^D\ell]} \left[  \log\left(\frac{M}{\delta} \right) + [2(r+r_0+k_0)\ell]\log\left( \frac{\gamma}{2}\right) + \log(q) + \log\left( \frac{\gamma}{t_\epsilon}\right) \right] \\
    &= O\left( \log\bigg(\frac{n}{\delta}\bigg)\log\bigg(\frac{M}{\delta}\bigg)2^{\polylog(f(n)/\epsilon)} \right).
\end{align*}
\end{theorem}

\begin{proof}
    We follow the reasoning of \cite[Theorem D.3]{Onorati_Rouze_Stilch_Franca_Watson_2023}:
    adapting the proof of \cref{Prop:Approximating_Observables_General}, it is clear that with probability  $$1-\delta:=1-Me^{-N\frac{1}{q}(\gamma/2)^{m_r}(\gamma/t_{\epsilon})+m_r\log(2/\gamma)+\log(t_\epsilon/\gamma)+\log q}$$
each cube is visited at least $q$ times. 
Conditioned on that event, and choosing $q$ such that \Cref{ttimes} holds, we have that with probability $1-\delta'$, and using the slow local mixing bound \cref{Corollary:Slow_Mixing_Localised_Lindbladian_All_Times}
\begin{align*}
    |f_O(\L(x),t) -\tilde{f}_O(\L(x),t) | &\leq   C_1(n) e^{-r/2\xi} \norm{O_A} +  \gamma\norm{O_A}J(\ell+1)(2r)^D + f(n)\norm{O_A}e^{-\gamma' t_\epsilon}.
\end{align*}
We divide up the error budget between these three terms equally such that:
\begin{align*}
    C_1(n) e^{-r/2\xi}, \  \gamma J(\ell+1)(2r)^D,\ f(n)e^{-\gamma' t_\epsilon} \leq \epsilon/3
\end{align*}
are all equal.
Thus we need to set 
\begin{align*}
    t_\epsilon &= \frac{1}{\gamma'}\log\left(\frac{3f(n)}{\epsilon}  \right) \\
     r &= 2\xi \log\left( \frac{3C_1(n)}{\epsilon} \right) = O\left(\log\left (\frac{f(n)}{\epsilon} \right)\right) \\ 
    \gamma &= \frac{\epsilon}{3J(\ell+1)(2(r+k_0))^D} = O\left( \frac{\epsilon}{\log^D(f(n))} \right).
\end{align*}
\end{proof}

\end{document}